\documentclass[graybox,numbook,envcountsame]{svmult}

\usepackage{amsmath, amsfonts, amssymb, url, eucal}
\usepackage{color,bbm,interval,mathtools}
\usepackage{type1cm}        
\usepackage{makeidx}         
\usepackage{graphicx}        
\usepackage{multicol}        
\usepackage[bottom]{footmisc}
\usepackage{newtxtext}       %
\usepackage[labelfont=bf,font=footnotesize]{caption}

\makeindex             

\DeclareMathAlphabet{\mathcal}{OMS}{cmsy}{m}{n}
\SetMathAlphabet{\mathcal}{bold}{OMS}{cmsy}{b}{n}
\usepackage{mathrsfs}

\usepackage[hyperfootnotes=false,pagebackref,colorlinks=true,urlcolor=blue,linkcolor=blue,citecolor=blue]{hyperref}
\renewcommand*{\backref}[1]{}\renewcommand*{\backrefalt}[4]{\ifcase #1 (Not cited.)\or (Cited on page~#2.)\else (Cited on pages~#2.)\fi}

\renewcommand{\a}{\alpha}
\renewcommand{\b}{\beta}

\newcommand{\g}{\gamma}

\renewcommand{\d}{\delta}

\renewcommand{\l}{\lambda}
\renewcommand{\L}{\Lambda}
\newcommand{\z}{w}			

\newcommand{\p}{\partial}
\newcommand{\om}{\omega}

\newcommand {\GS}{\mathfrak S}
\newcommand{\bv}{\mathbf v}
\newcommand {\bx}{\mathbf x}
\newcommand {\be}{\mathbf e}
\newcommand {\bz}{\mathbf z}
\newcommand {\by}{\mathbf y}
\newcommand {\bn}{\mathbf n}
\newcommand {\bxi}{\boldsymbol\xi}
\newcommand{\lu}{\langle}
\newcommand{\ru}{\rangle}

\newcommand{\CN}{\mathcal N}
\newcommand{\CB}{\mathcal B}
\newcommand{\CL}{\mathcal L}
\newcommand{\CT}{\mathcal T}
\newcommand{\CH}{\mathcal H}
\newcommand{\CA}{\mathcal A}
\newcommand{\CD}{\mathcal D}

\newcommand{\plainC}[1]{\textup{{\textsf{C}}}^{#1}}
\newcommand{\plainL}[1]{\textup{{\textsf{L}}}^{#1}}

\DeclareMathOperator{\tr}{{tr}}
\DeclareMathOperator{\op}{{Op}}

\spnewtheorem{thm}{Theorem}[section]{\bfseries}{\itshape}
\spnewtheorem{cor}[thm]{Corollary}{\bfseries}{\itshape}
\spnewtheorem{lem}[thm]{Lemma}{\bfseries}{\itshape}
\spnewtheorem{prop}[thm]{Proposition}{\bfseries}{\itshape}
\spnewtheorem{cond}[thm]{Condition}{\bfseries}{\itshape}
\spnewtheorem{rem}[thm]{Remark}{\bfseries}{\rmfamily}
 
\newcommand{\mfr}[2]{{\textstyle\frac{#1}{#2}}}
\newcommand{\N}{\mathbbm N}
\newcommand{\R}{ \mathbbm R}
\renewcommand{\i}{{\mathrm i}}
\newcommand{\ceq}{\coloneqq}

\begin{document}

\title*{R\'enyi entropies of the free Fermi gas\\ in multi-dimensional space at high temperature }
\author{Hajo~Leschke, Alexander V.~Sobolev, Wolfgang~Spitzer}
\authorrunning{H.~Leschke, A.\,V.~Sobolev, W.~Spitzer}
\institute{
	Hajo~Leschke \at Institut f\"ur Theoretische Physik, Universit\"at Erlangen--N\"urnberg, Staudtstra\ss e 7, 91058 Erlangen, Germany, \email{hajo.leschke@physik.uni-erlangen.de}
\and
	Alexander V.~Sobolev \at Department of Mathematics, University College London, Gower Street, London WC1E 6BT, United Kingdom, \email{a.sobolev@ucl.ac.uk}
\and
	Wolfgang~Spitzer \at Fakult\"at f\"ur Mathematik und Informatik, FernUniversit\"at in Hagen, Universit\"atsstra\ss e 1, 58097 Hagen, Germany, \email{wolfgang.spitzer@fernuni-hagen.de}
}
\numberwithin{equation}{section}
\maketitle
\centerline{\fbox{\emph{In memory of Harold Widom (1932--2021)}}}
\bigskip\bigskip
%
%
\abstract{
We study the local and (bipartite) entanglement R\'enyi entropies of the free Fermi gas in multi-dimensional Euclidean space $\R^d$ in thermal equilibrium. We prove positivity of the entanglement entropies with R\'enyi index $\g \leq1$ for all temperatures $T>0$. Furthermore, for general $\g > 0$  we establish the asymptotics of the entropies for large $T$ and large scaling parameter $\a >0$ for two different regimes -- for fixed chemical potential $\mu\in\R$ and also for fixed particle density $\rho>0$. In particular, we thereby provide the last remaining building block for complete proofs of 
our low- and high-temperature results presented (for $\g = 1$) in J.~Phys. A: Math.~Theor. \textbf{49}, 30LT04 (2016) [Corrigendum: \textbf{50}, 129501 (2017)], 
but being supported there only by the basic proof ideas.
} 



\bigskip\noindent {\footnotesize Date of this version: 29 April 2023,\quad Preprint: \,\href{https://arxiv.org/abs/2201.11087}{arXiv: 2201.11087} [math-ph]}

\noindent {\footnotesize Print (in slightly different form): in the book series\\ \emph{Operator Theory: Advances and Applications} vol. 289, pp. 477-508, 2022}
\newpage

\section*{Contents}
\setcounter{minitocdepth}{2}
\dominitoc

\section{Introduction}

This section briefly describes the physical background, introduces the most important mathematical definitions, and provides a summary of the main results.

\subsection{Physical background}

Over the last decades the so-called entanglement entropy (EE) has become a useful single-number quantifier of non-classical correlations between subsystems of a composite quantum-mechanical system \cite{Amico08,HHHH09}. For example, one may imagine a macroscopically large system consisting of a huge number of particles in the state of thermal equilibrium at some (absolute) temperature $T\geq 0$. All the particles inside some bounded spatial region $\L$ may then be considered to constitute one subsystem and the particles outside of $\L$ another one. The corresponding EE, more precisely the spatially bipartite thermal EE, now quantifies, to some extent, how strongly these two subsystems are correlated ``across the interface'' between $\L$ and its complement. 

In the simplified situation where the particles do not dynamically interact with each other, such as in the ideal gas or, slightly more general, in the free gas, all possible correlations are entirely due to either the Bose--Einstein or the \mbox{(Pauli--)}Fermi--Dirac statistics by the assumed indistinguishability of the (point-like and spinless) particles.The present study is devoted to the latter case. Accordingly, we consider the free Fermi gas \cite{Fermi26,Balian92,Bratteli97} infinitely extended in the Euclidean space $\R^d$ of an arbitrary dimension $d\geq1$. Although the fermions neither interact with each other nor with any externally applied field, their EE remains a complicated function of the region $\L\subset\R^d$ which is difficult to study by analytical methods. In general one can only hope for estimates and asymptotic results for its (physically interesting) growth when $\L$ is replaced with $\a\L$ where the scaling parameter $\a>0$ becomes large. A decisive progress towards the understanding of the growth of the EE at $T=0$, in other words of the ground-state EE, is due to Gioev and Klich \cite{G06,GK06}. They observed, remarkably enough, that this growth is related to a conjecture of Harold Widom \cite{W82,W90} about the quasi-classical Szeg\H{o}-type asymptotics for traces of (smooth) functions of multi-dimensional versions of truncated Wiener--Hopf operators with discontinuous symbols. After Widom's conjecture had been proved by one of us \cite{Sobolev2013,S15} the gate stood open to confirm \cite{LSS14} the precise (leading) large-scale growth conjectured in \cite{GK06} and, in addition, to establish its extension from the von Neumann EE to the whole one-parameter family of (quantum) R\'enyi EE's.

In the present study we only consider the case of a true thermal state characterized by a \emph{strictly} positive temperature $T>0$ (and a chemical potential $\mu \in \R$ or, equivalently, a spatial particle-number density $\rho > 0$). On the one hand, the case $T>0$ is simpler, because the \emph{Fermi function} $E\mapsto 1/\big(1+\exp(E/T)\big)$ on the real line $\R$ is smooth in contrast to its zero-temperature limit, the Heaviside unit-step function. On the other hand, a reasonable definition of the EE is more complicated, because the thermal state is not a pure state, see \eqref{def:EE} below. Nevertheless, due to the presence of the ``smoothing parameter'' $T>0$ the leading asymptotic growth of the EE as $\a\to\infty$, is determined by an asymptotic coefficient again going back to Widom \cite{W60,Widom1985}, see also \cite{Rocca84,BuBu91}. We introduce it in \eqref{cbd:eq} below and denote it by $\CB$.

From a physical point of view it is interesting to study the scaling asymptotics $\a\to\infty$ as the temperature $T$ varies. The emerging double asymptotics of the EE and of the coefficient $\CB$ are not simple to analyze and hard to guess by heuristic arguments.
For low temperatures, that is, small $T>0$, this analysis has been performed in \cite{LeSpSo3} for $d=1$ and in \cite{S17} for $d\geq 2$ yielding a result consistent with that for $T=0$ in \cite{GK06,LSS14}. 

At high temperatures quantum effects become weaker and the free Fermi gas should exhibit properties of the corresponding classical free gas without correlations (for fixed particle density). In particular, the ideal Fermi gas \cite{Fermi26,Balian92} should behave like the Maxwell--Boltzmann gas, the time-honored ``germ cell'' of modern statistical mechanics. Hence the main purpose of our study is to determine the precise two-parameter asymptotics of the EE as $\a\to\infty$ and $T\to\infty$.   

\subsection{Pseudo-differential operators and entropies}

At first we introduce the translation invariant pseudo-differential operator\footnote{In \cite{LeSpSo3,Sobolev2019} the right-hand side of \eqref{eq:opa} is mistakenly multiplied by $(2\pi)^{d/2}$\,.}
\begin{equation}\label{eq:opa}
	\bigl(\op_\a(a) u\bigr)(\bx)
	\ceq\frac{\a^{d}}{(2\pi)^d}
	\iint_{\R^d\times \R^{d}} \mathrm{e}^{\i\a\bxi\boldsymbol{\cdot}(\bx-\by)} a(\bxi) 
	u(\by) \,\mathrm{d}\bxi \mathrm{d}\by\,,\quad \bx\in\R^d\,.
\end{equation} 
Here the smooth real-valued function $a$ is its underlying \emph{symbol}, $u$ is an arbitrary complex-valued Schwartz function, $\i$ denotes the imaginary unit, and $\a>0$ is the scaling parameter. Informally, one may think of $\op_\a(a)$ as the function $a(-(\i/\a)\boldsymbol{\nabla})$ of the gradient operator $\boldsymbol{\nabla}\ceq (\p_{x_1},\p_{x_2},\dots,\p_{x_d})$, that is, the vector of partial derivatives with respect to ${\bx}=(x_1, x_2,\dots, x_d)$. 

The main role will be played by the \emph{truncated Wiener--Hopf operator}  
\begin{equation*}
	W_\a(a, \L)\ceq\chi_\L \op_\a(a) \chi_\L\,,
\end{equation*}
where $\chi_\L$ is the (multiplication operator corresponding to the) indicator function of the ``truncating'' open set $\L\subset\R^d$. 
Clearly, for a bounded symbol $a$ the operators $\op_\a(a)$ and $W_\a(a; \L)$ are bounded on the Hilbert space $\plainL2(\R^d)$. Given a \emph{test function} $f: \R\to\R$, we are interested in the operator $f\big(W_\a(a, \L)\big) $ and in the operator difference 
\begin{equation}\label{Dalpha:eq}
	D_\a(a, \L; f) \ceq \chi_\L f\big(W_\a(a, \L)\big) \chi_\L - W_\a(f\circ a,\L)\,,
\end{equation} 
where the symbol $f\circ a$ is the composition of $f$ and $a$ defined by $(f\circ a)(\bxi)\ceq f\big(a(\bxi)\big)$.

For a bounded $\L$ and suitable $a$ and $f$  both operators on the right-hand side of \eqref{Dalpha:eq} belong individually to the trace class. Remarkably, its difference does so even for a large class of unbounded $\L$, see Condition \ref{domain:cond} and Proposition \ref{fixedT:prop} below. Our analysis of the scaling behavior of the entropies will be based on the asymptotics for the trace of $D_\a(a, \L; f)$ as $\a\to\infty$. The reciprocal parameter $\a^{-1}$ can be naturally viewed as the Planck constant, and hence the limit $\a\to\infty$ can be regarded as the quasi-classical limit. By a straightforward change of variables the operator \eqref{Dalpha:eq} is seen to be unitarily equivalent to $D_1(a, \a\L; f)$, so that  $\a\to\infty$ can be interpreted also as a spatial large-scale limit. In our large-scale applications it is either $\a$ itself or a certain combination of $\a$ with the temperature $T$ that will become large.

The macroscopic thermal equilibrium state of the free Fermi gas depends, first of all, on its (classical) single-particle Hamiltonian  $h:\R^d\to \R$, sometimes also called the energy-momentum (dispersion) relation. The minimal conditions that we impose on $h$ are as follows. We assume that $h$ is smooth in the sense that $h\in\plainC\infty(\R^d)$ and that it satisfies the bounds
\begin{align}\label{eq:hin}
	\big|\p_{\bxi}^n h(\bxi)\big|\lesssim |\bxi|^{2m}\quad \text{for all} \,\,  n \in \N_0^d\,,
	\textup{ and }\, h(\bxi)\gtrsim |\bxi|^{2m}\,\, \textup{for}\,\, |\bxi|\gtrsim1\,,
\end{align}
for some constant $m>0$. Here $\N_0\ceq \N\cup \{0\}$ and the notations $\p_{\bxi}^n$, $\lesssim$\,,\, $\gtrsim$\, are defined at the end of the Introduction.

For given $h$, temperature $T>0$, and chemical potential $\mu\in\R$  we introduce the \emph{Fermi symbol} as the composition of the Fermi function and the ``shifted'' Hamiltonian $h-\mu$ by
\begin{equation}\label{positiveT:eq} 
	a_{T, \mu}(\bxi) \ceq \frac{1}{1+ \exp{\big(\frac{h(\bxi) - \mu}{T}\big)}}\,,\quad \bxi \in\R^d\,.
\end{equation}

Next, we introduce the (bounded and continuous)  \emph{entropy function} $\eta_\g:\R\to \interval{0}{\ln(2)}$ for each \emph{R\'enyi index} $\g>0$. If $t\not\in {]0, 1[}$, we set $\eta_\g(t)\ceq0$ for any $\g>0$. If $t\in {]0, 1[)}$ and $\g\not=1$ we define
\begin{equation}\label{eta1:eq}
\eta_\g(t) \ceq \frac{1}{1-\g} \ln\big[t^\g + (1-t)^\g\big]\,.
\end{equation}
The von Neumann case, $\g=1$, is then given by the point-wise limit
\begin{equation}\label{eq:etaln}
\eta_1(t)\ceq \lim_{\g\to1} \eta_\g(t) =-t \ln(t) -(1-t)\ln(1-t)\,.	
\end{equation}

%
%
%

%
Finally, we define the (R\'enyi) \emph{local entropy} associated with a \emph{bounded} region $\L$ as the trace 
\begin{align}\label{def:local E}
	\mathrm{S}_{\g}(T, \mu; \L) \ceq \tr \eta_\g\big(W_1(a_{T, \mu}, \L)\big)\geq 0\,,
\end{align}
and the (R\'enyi) \emph{entanglement entropy} (EE) for the bipartition $\R^d = \L \cup (\R^d\setminus\L)$   as
\begin{align}\label{def:EE}
	\mathrm{H}_{\g}(T, \mu; \L) \ceq \tr D_1(a_{T,\mu}, \L; \eta_\g) + \tr D_1(a_{T,\mu},\R^d\setminus\L;\eta_\g)\,.
\end{align}
This definition is motivated by the notion of mutual information, see e.g. \cite[Eq.\,(9)]{Amico08}. The conditions \eqref{eq:hin} guarantee that these entropies are well-defined, see the paragraph after Proposition \ref{fixedT:prop}. In formula \eqref{def:EE} either $\L$ or its complement $\R^d\setminus\L$ is assumed to be bounded, see Section \ref{sect:prelim} for details. It is useful to observe, as in \cite[Eq.\,(26)]{LeSpSo2}, that the local entropy \eqref{def:local E} can be expressed as
\begin{align}\label{eq:local_ED}
	\mathrm{S}_\g(T, \mu; \L) = s_\g(T, \mu)\,|\L| + \tr D_1(a_{T, \mu}, \L; \eta_\g)\,.
\end{align}
Here $|\L|$ is the volume (Lebesgue measure) of the bounded region $\L$ and 
\begin{align}\label{density}
	s_\g(T, \mu) \ceq 
	\frac{\tr W_1(\eta_\g\circ a_{T, \mu}, \L)}{|\L|}
	= \frac{1}{(2\pi)^{d}} \int_{\R^d} \eta_\g\big(a_{T,\mu}(\bxi)\big)\, \mathrm{d}\bxi
\end{align} 
is the spatial \emph{entropy density}. 
For $\g=1$ it is the usual thermal entropy density \cite{Balian92,Bratteli97}.
%

\subsection{Summary of the main results}

Our first result is of non-asymptotic nature. In Section \ref{sect:positive} we explore concavity properties of the function $\eta_\g$.  First we notice that $\eta_\g$ is concave on the unit interval ${[0, 1]}$ for all $\g\in{]0, 2]}$, so that for a bounded $\L$ one can use a Jensen-type trace inequality to establish a lower bound for the local entropy \eqref{def:local E} in terms of the entropy density \eqref{density}, see Theorem~\ref{cor:local}. For the EE \eqref{def:EE} this argument is not applicable, but we observe that $\eta_\g$ for $\g\in{]0,1]}$ is even \emph{operator concave} so that the Davis operator inequality \eqref{eq:davis} implies the positivity of the EE for $\g\in{]0,1]}$ (including the most important case $\g=1$ corresponding to the von Neumann EE), see Theorem \ref{thm:EE positive}. We do not know whether the EE is positive for $\g\in{]1,2]}$. Later on however, we will see that the EE for $\g\in{]1, 2]}$ is positive at least for large $T$, see Remarks~\ref{item:negative} and \ref{item:afixed} in Subsection~\ref{subsect:hta}.
 
The other main objective of the present paper is to study the asymptotics of ${\mathrm H}_{\g}(T, \mu; \a\L)$ as $\a\to\infty$ and $T\to\infty$. For this, we have to impose conditions on the Hamiltonian $h$ stronger than those in \eqref{eq:hin}, see Subsection \ref{subsect:h}. In particular, $h$ should be asymptotically homogeneous as $|\bxi|\to\infty$. We distinguish between two high-temperature cases: the case of a constant chemical potential $\mu\in\R$ and the case where the mean particle density
\begin{align}\label{eq:density}
	\varrho(T,\mu) \ceq \frac{1}{(2\pi)^d}\int_{\R^d} a_{T, \mu}(\bxi)\, \mathrm{d}\bxi\,, 
\end{align}
being finite by \eqref{eq:hin}, is (asymptotically) fixed to a prescribed constant  $\rho >0$, as $T$ becomes large. The latter case implies that the chemical potential effectively becomes a function of $T$ (and $\rho$) in the sense that 
\begin{align}\label{rhofixed}
 \varrho(T, \mu_\rho(T)) \to \rho \quad \textup{as}\quad T\to\infty\,. 
\end{align}
This corresponds to the quasi-classical limit of the free Fermi gas at fixed particle density. Indeed, since the so-called integrated density of states 
\begin{align}\label{eq:dos}
	\mathcal N(T)\ceq\frac{1}{(2\pi)^d}\int_{h(\bxi)< T} \mathrm{d}\bxi
\end{align}
of the single-particle Hamiltonian $h$, evaluated at $T$, tends to infinity as $T\to\infty$, one has $\rho/\CN(T)\to 0$. In physical terms, the spatial density of the number of particles is much smaller than that of the number of occupiable single-particle states with (eigen)energies below $T$, when $T$ is sufficiently large. Therefore the restrictions by the Pauli exclusion principle are negligible.

Our results on the high-temperature scaling asymptotics are presented in Theorem \ref{thm:fixedmu} (for constant $\mu$) and in Theorem \ref{thm:fixedrho} (for constant $\rho$). We postpone the discussion of these theorems until Section \ref{sect:hta}. Here we only mention two important facts: a) in both regimes the asymptotics still hold if $T\to\infty$ and the scaling parameter $\a$ is fixed, e.g. $\a = 1$, b) the EE with R\'enyi index $\g>2$ becomes \emph{negative} for fixed particle density at high temperature; the same happens for fixed chemical potential at large $\g$. This suggests that values $\g>2$ are only of limited physical interest. 
  
The main asymptotic orders as $\a\to\infty$ and $T\to\infty$ have been presented (without precise pre-factors and proofs) in \cite{LeSpSo2}, in the cases of fixed $\mu$ and fixed $\rho$, for $\g=1$ and all $d\ge 1$. Here we provide complete proofs for all $\g>0$, but concentrate on the multi-dimensional case $d\ge 2$. The case $d=1$ is not considered for lack of space. 

\medskip

The paper is organized as follows. In Section \ref{sect:prelim} we present the basic information such as our conditions on the truncating region and the definition of the asymptotic coefficient $\CB$. Section 2 also contains the results, borrowed mostly from \cite{Sobolev2019}, that are used throughout the paper. In Section \ref{sect:positive} we study the concavity of the function $\eta_\g$ and investigate the positivity of the corresponding entanglement entropy. In Section \ref{sect:qc} we derive elementary trace-class bounds for some abstract bounded (self-adjoint) operators. These bounds are based on estimates for \emph{quasi-commutators} of the form $f(A) J - J f(B)$ with bounded $J$, bounded self-adjoint $A$, $B$, and suitable functions $f$. By using such bounds in Section \ref{sect:model} we obtain the large $\a$ and $T$ asymptotics for the trace of the operator $D_\a(p_T, \L; f)$ with a symbol $p_T$ modeling the Fermi symbol \eqref{positiveT:eq} in the fixed $\mu$ or fixed $\rho$ regimes. In Section \ref{sect:hta} we collect our results on the high-temperature asymptotics for the EE \eqref{def:EE}, 
see Theorems \ref{thm:fixedmu} and \ref{thm:fixedrho}. Their proofs are directly based on the formulas obtained in Section \ref{sect:model}. In Subsection \ref{subsect:hta} we also comment on the asymptotics of the local entropy \eqref{def:local E}. The Appendix contains a short calculation clarifying the structure of the Fermi symbol when the mean particle density is fixed as $T\to\infty$.

Throughout the paper we adopt the following standard notations. For two positive numbers (or functions) $X$ and $Y$, possibly depending on parameters, we write $X\lesssim Y$ (or $Y\gtrsim X$) if $X\le C Y$ with some constant $C\geq 0$ independent of those parameters. If $X\lesssim Y$ and $X\gtrsim Y$, then we write $X\asymp Y$. To avoid confusion we often make explicit comments on the nature of the (implicit) constants in the bounds. For multiple partial derivatives we use the notation $\p_{\bxi}^n \ceq \p_{\xi_1}^{n_1}\p_{\xi_2}^{n_2}\cdots\p_{\xi_d}^{n_d}$ for a vector $\bxi\in\R^d$ and a multi-index $n = (n_1, n_2, \dots, n_d)\in\N_0^d$ of order $|n|\ceq n_1+n_2+ \cdots +n_d$. By $B(\bz, R)$ we mean the open ball in $\R^d$ with center $\bz\in\R^d$ and radius $R>0$. We also use the weight function $\lu \bv\ru \ceq \sqrt{1+|\bv|^2}$ for any vector $\bv\in\R^d$.

The notation $\GS_p$, $p\in {]0,\infty[},$ is used for the Schatten--von Neumann classes of compact operators on a complex separable Hilbert space $\CH$, see e.g.~\cite[Chapter 11]{BS}. By definition, the operator $A$ belongs to $\GS_p$ if $\|A\|_p \ceq (\tr(A^*A)^{p/2})^{1/p}<\infty$. The functional $\|\cdot\|_p$ on $\GS_p$ is a norm if $p\ge 1$ and a quasi-norm if $p<1$. Apart from Section \ref{sect:qc}, where the space $\CH$ is arbitrary, we assume that $\CH = \plainL2(\R^d)$.

\bigskip
 
We dedicate this paper to the memory of Harold Widom (1932--2021). His ground-breaking results on the asymptotic expansions for traces of pseudo-differential operators have been highly influential to many researchers including us. Without his results the present contribution and our previous ones to the study of fermionic entanglement entropy would have been unthinkable. We are deeply indebted to Widom's ingenious insights. All three of us had the honor and pleasure of meeting him at a memorable workshop in March 2017 hosted by the American Institute of Mathematics (AIM) in San Jose, California, USA.

\begin{figure}[h!]
\begin{center}
\rotatebox{90}{{\scriptsize Photographed by H. Leschke}}
		\includegraphics[scale=0.3]{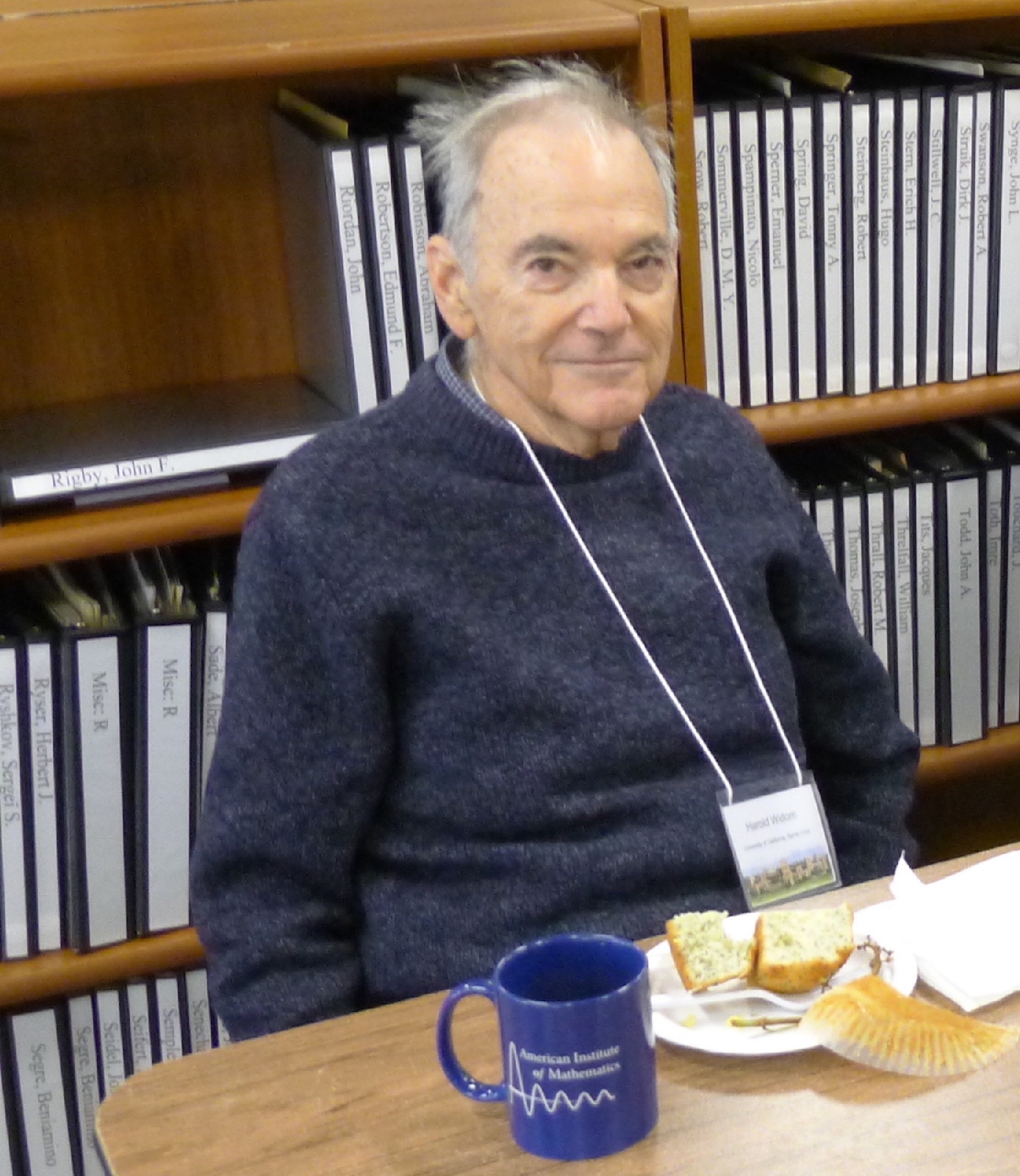}
\end{center}
\captionsetup{justification=centering}
\caption{Harold Widom on 28 March 2017 at the AIM in San Jose, CA}
\end{figure}

\section{Basic definitions and basic facts}\label{sect:prelim}

In this section we collect some definitions and facts from \cite{Sobolev2019} concerning the trace of \eqref{Dalpha:eq} and its asymptotic evaluation. They are instrumental in the proof of our main asymptotic results corresponding to $a=a_{T,\mu}$ and $f=\eta_\g$. Throughout the rest of the paper we always assume $d\ge 2$ for the spatial dimension.

\subsection{Conditions on the truncating region $\L$} 

We call an open and connected set $\L\subset \R^d$ a Lipschitz domain, if it can be described locally as the epigraph of a Lipschitz function, see \cite{Sobolev2014,S15} for details. We call $\L$ a \emph{Lipschitz region}, if it is a union of finitely many Lipschitz domains such that their closures are pair-wise disjoint. From now on we always assume that $\L$ satisfies the following condition. Nevertheless, for convenience we will often mention it.
\begin{cond}\label{domain:cond}
\begin{enumerate}
	\item The set $\L\subset \R^d$ is a Lipschitz region, and  either $\L$ or $			   \R^d\setminus\L$ is bounded. 
	\item The boundary (surface) $\p\L$ is piece-wise $\plainC1$-smooth.
\end{enumerate}
\end{cond}
\noindent We note that $\L$ and $\R^d\setminus\overline{\L}$ satisfy Condition \ref{domain:cond} simultaneously.

\subsection{The asymptotic coefficient  $\CB$  and its basic properties} 

We assume the real-valued symbol to be smooth in the sense that $a\in\plainC\infty(\R^d)$ and satisfies the decay condition
\begin{align}\label{abounds:eq}
	\big|\p_{\bxi}^n a(\bxi)\big|  \lesssim \lu\bxi\ru^{-\b}\quad \text{ with some constant } \quad \b > d\,, 
\end{align}
for all $\bxi\in\R^d$ and all $n \in \N_0^d$ with some implicit constants that may depend on $n$. 

Before stating the leading asymptotic formula for $\tr D_\a(a, \L; f)$ as $\a\to\infty$, we need to introduce the corresponding asymptotic coefficient. For a function $f:\R\to \R$ and any $u, v\in\R$ we consider the integral 
\begin{equation}\label{U:eq}
	U(u, v; f)\, 
	\ceq\int_0^1 \frac{f\bigl(tu + (1-t)v\bigr) 
		- \big[tf(u) + (1-t)f(v)\big]}{t(1-t)} \,\mathrm{d} t\,. 
\end{equation} 
It is well-defined for any H\"older continuous $f$. And it is positive if $f$ is also concave. 

For every unit vector $\be\in\R^d$  we define a functional of the symbol $a$ by the principal-value integral:
\begin{align}\label{ad:eq}
	\CA(a, \be; f) \ceq\frac{1}{8\pi^2}\lim\limits_{\varepsilon\downarrow 0}
	\int_{\R^d} \int_{|t|>\varepsilon} 
	\frac{U\bigl(a(\bxi), a(\bxi + t\be); f\bigr)}{t^2}\,\mathrm{d}t\,\mathrm{d}\bxi\,.
\end{align}
Finally we define the main asymptotic coefficient by
\begin{align}\label{cbd:eq}
	\CB(a, \p\L;  f) \ceq \frac{1}{(2\pi)^{d-1}}
	\int_{\p\L} \CA(a, \bn_\bx; f) \, \sigma(\mathrm{d}\bx)\,,
\end{align} 
where $\bn_\bx$ is the (unit outward) normal vector at the point $\bx\in\p\L$ and $\sigma$ is the canonical $(d-1)$-dimensional area measure on the boundary surface $\p\L$.

For a function $f\in\plainC2(\R)$ with bounded second derivative and for a symbol $a$ obeying the condition \eqref{abounds:eq} the integral \eqref{ad:eq} exists in the usual sense and is bounded uniformly in $\be$. Hence \eqref{cbd:eq} is also finite, see \cite[Section 3]{Sobolev2019}. However, in order to accommodate the entropy function $\eta_\g$ we allow for test functions being non-smooth in the sense of the following condition.

\begin{cond}\label{cond:f}
The function $f$ is in $\plainC2(\R\setminus \CT)\cap\plainC{}(\R)$, where $\CT\ceq\{ t_1, t_2, \dots, t_N\}$ is a finite set of its singular points. Moreover, for some $\d >0$ and all $R>0$ the function $f\eqqcolon f^{(0)}$ and its first two derivatives satisfy the bounds
\begin{equation}\label{eq:fbound}
	\big|f^{(k)}(t)\big| \lesssim \sum_{u\in\CT}|t-u|^{\d-k}\,,\quad k\in \{ 0, 1, 2\}\,,
\end{equation} 
for all $t\in[-R, R]\setminus \CT$ with an implicit constant that may depend on $R$.
\end{cond}

\noindent Under this condition $\CB(a, \p\L; f)$ is finite: 

\begin{prop}\label{bddef:prop} \cite[Corollary 3.4]{Sobolev2019}
Let the set $\L$ satisfy Condition \ref{domain:cond} and let the function $f$ satisfy Condition \ref{cond:f}  with $\d>0$. Finally, let the symbol $a$ satisfy \eqref{abounds:eq}, but this time with some $\b > d\,\max\{1, 1/\d \}$. Then the coefficient $\CB(a, \p\L; f)$ is finite.  
\end{prop}

\noindent We point out a few simple properties of this coefficient. 

\begin{rem}\label{rem:signb}
\begin{enumerate}
	\item Since $\CA(a, \be; f) = \CA(a, -\be; f)$, the coefficients $\CB$ for the regions $\L$ and $\R^d\setminus\overline{\L}$ coincide. 
	\item By the definition \eqref{U:eq}, the coefficient \eqref{cbd:eq} is positive for concave functions $f$ and negative for convex ones. For example, the function $\eta_\g$ with $\g\in {]0, 2]}$ is concave on the interval ${[0, 1]}$ (see Lemma \ref{lem:ren}) so that $\CB(a, \p\L; \eta_\g)\ge 0$ for symbols $a$ taking values only in ${[0, 1]}$, like the Fermi symbol $a_{T,\mu}$. 
	\item If the symbol $a$ is spherically symmetric (for example, by spherical symmetry of the Hamiltonian $h$ in $a_{T,\mu}$), then the surface area $|\partial\L|$ factors out of $\CB(a, \p\L; f)$. Nevertheless, the remaining integral is still hard to compute for general $a$ and $f$. See, however, Remark \ref{item:negative} in Subsection \ref{subsect:hta} for Gaussian $a$ and quadratic $f$.
\end{enumerate}
\end{rem}

A less obvious property of the coefficient $\CB(a, \p\L; f)$ is its continuity in the symbol $a$. Since it is important for our purposes, we state a corresponding result. In the next and subsequent assertions we consider a one-parameter family of symbols $\{a_0, a_\l\}$, $\l>0$, all of them satisfying \eqref{abounds:eq} with some $\b > d\,\max\{1,\d^{-1}\}$, uniformly in $\l$, and point-wise convergence $a_\l\to a_0$ as $\l\downarrow 0$. 

\begin{prop}\label{prop:contb} \cite[Corollary 3.5]{Sobolev2019}
Let the set $\L$ and the function $f$ be as in Proposition \ref{bddef:prop}. Then 
\begin{align}\label{cbconvf:eq}
	\CB(a_\l, \p\L; f)\to \CB(a_0, \p\L; f)\,\quad \textup{as}\,\,\, \l\downarrow 0\,.
\end{align}
\end{prop} 

\subsection{The asymptotics for $\tr D_\a(a, \L; f)$ as $\a\to\infty$}\label{subsect:asym}
%

Now we are in a position to state the required asymptotic facts.  

\begin{prop}\label{fixedT:prop}\cite[Theorem 2.3]{Sobolev2019}
Let the set $\L$, the function $f$, and the symbol $a$ be as in Proposition \ref{bddef:prop}. Then the operator $D_\a(a, \L; f)$ is of trace class and 
\begin{align*} 
	\lim\limits_{\a\to\infty} \a^{1-d}\tr D_\a(a, \L; f) = \CB(a, \p\L; f)\,.
\end{align*}
This limit is uniform in the class of symbols $a$ that satisfy \eqref{abounds:eq} with the same implicit constants.
\end{prop}

Proposition \ref{fixedT:prop} ensures the existence of the entropies \eqref{def:local E} and \eqref{def:EE}. In fact, assume that $\L$ satisfies Condition \ref{domain:cond} and that the Hamiltonian $h$ in the Fermi symbol $a_{T,\mu}$ of \eqref{positiveT:eq} is as specified in \eqref{eq:hin}. Then $a_{T,\mu}$  satisfies \eqref{abounds:eq} for all $T>0$ and $\mu\in\R$. Moreover, the function $\eta_\g$ satisfies Condition \ref{cond:f} for all $\g >0$ with arbitrary $\d < \min\{1, \g\}$ and the set $\CT = \{0, 1\}$. Thus, due to Proposition \ref{fixedT:prop}, the operators $D_1(a_{T, \mu}, \L; \eta_\g)$ and $D_1(a_{T, \mu}, \R^d\setminus\L; \eta_\g)$ are of trace class, so that the entanglement entropy \eqref{def:EE} is finite. If we additionally assume that $\L$ is bounded, then by \eqref{eq:local_ED} also the local entropy \eqref{def:local E} is finite. 

Proposition \ref{fixedT:prop} was also used in \cite{Sobolev2019} to determine the scaling asymptotics for the entanglement entropy at fixed temperature. To study the high-temperature regime, we need the continuity of this result in the symbol $a$:

\begin{cor}\label{prop:uni}
Let the set $\L$ and the function $f$ be as in Proposition \ref{bddef:prop}. Then 
\begin{align}\label{eq:uni}
	\lim\,\a^{1-d} \tr D_\a(a_\l, \L; f) = \CB(a_0, \p\L; f)\,,
\end{align} 
where the limits $\a\to\infty$ and $\l\to 0$ are taken simultaneously.
\end{cor}

\begin{proof}
According to Proposition \ref{fixedT:prop}, 
\begin{align*}
	\lim\limits_{\a\to\infty} \,\a^{1-d} \tr D_\a(a_\l, \L; f) = \CB(a_\l, \p\L; f)\,,
\end{align*}
uniformly in $\l$. Together with \eqref{cbconvf:eq} this leads to \eqref{eq:uni}. 
\end{proof}

\noindent The next two propositions describe the asymptotics for ``small'' symbols. 

\begin{prop}\label{prop:homo}\cite[Theorem 2.5]{Sobolev2019}  
Let the set $\L$ satisfy Condition \ref{domain:cond} and let $f_0$ be the function defined  by $f_0(t) := M |t|^\d$ with real constants $M$ and $\d >0$. Finally, suppose that the function $f\in \plainC2(\R\setminus\{0\})$ satisfies the conditions
\begin{align}\label{almosthom:eq}
	\lim\limits_{t\to 0} |t|^{k-\d}\frac{\mathrm{d}^k}{\mathrm{d}t^k} \big(f(t) - f_0(t)\big) = 0\,, \quad k\in \{0, 1, 2\}\,.
\end{align} 
Then 
\begin{align*}
	\lim\limits_{\a\to\infty\, \l\to 0} \bigl(\a^{1-d} \l^{-\d}\tr D_\a(\l a_\l, \L; f)\bigr) = \CB(a_0, \p\L; f_0)\,.
\end{align*}
\end{prop}

In the next proposition we consider instead of the homogeneous function $f_0$ the function $\eta$ defined by
\begin{align}\label{eta:eq}
	\eta(t) \ceq -t\ln(|t|)\,, \quad  t\in\R\,,
\end{align}
which still leads to an asymptotically homogeneous behavior.  
 
\begin{prop}\label{prop:homolog} \cite[Theorem 2.6]{Sobolev2019}
Let the set $\L$ satisfy Condition \ref{domain:cond} and suppose that the function $f\in \plainC2(\R\setminus\{0\})$ 
satisfies the conditions 
\begin{align}\label{homoeta:eq}
	\lim\limits_{t\to 0} |t|^{k-1}\frac{\mathrm{d}^k}{\mathrm{d}t^k} \big(f(t) - \eta(t)\big) = 0\,,\quad k\in\{ 0, 1, 2\}\,.
\end{align}
Then 
\begin{align*}
	\lim\limits_{\a\to\infty\, \l\to 0} \bigl(\a^{1-d} \l^{-1}\tr D_\a(\l a_\l, \L; f)\bigr) = \CB(a_0, \p\L; \eta)\,.
\end{align*}
\end{prop}
We note that under any of the assumptions \eqref{almosthom:eq} and \eqref{homoeta:eq} the function $f$ satisfies Condition \ref{cond:f} with $\mathcal T = \{0\}$. For assumption \eqref{almosthom:eq} (resp. \eqref{homoeta:eq} ) the condition \eqref{eq:fbound} holds with the constant $\d$ from \eqref{almosthom:eq} (resp. arbitrary $\d <1$). 

The asymptotic results listed above are useful, but, as they stand, not directly applicable for our purposes. This is because our symbol of main interest, the Fermi symbol  \eqref{positiveT:eq}, depends on the two parameters $T$ and $\mu$, and in the course of our analysis in Section \ref{sect:model} we naturally come across certain ``effective'' symbols that do not satisfy conditions like \eqref{abounds:eq} uniformly in these parameters. However, we overcome this problem by considering a wider class of symbols, called \emph{multi-scale symbols} in \cite[Section 3]{LeSpSo3}.

\subsection{Multi-scale symbols}  

We consider symbols $a\in\plainC\infty(\R^d)$ for which there exist two continuous functions $\tau$ and $v$ on $\R^d$ with $\tau>0$, $v>0$, $v$ bounded, and such that
\begin{equation}\label{scales:eq}
	\big|\p_{\bxi}^k a(\bxi)\big|\lesssim  \tau(\bxi)^{-|k|} v(\bxi)\,,\ k \in \N_0^d,\quad \bxi\in\R^d\,,
\end{equation}
with implicit constants independent of $\bxi$. It is natural to call $\tau$ the \emph{scale (function)} and $v$ the \emph{amplitude (function)}. The scale $\tau$ is assumed to be globally Lipschitz continuous with Lipschitz constant $\CL<1$, that is,
\begin{equation}\label{Lip:eq}
	|\tau(\bxi) - \tau(\bxi')| \le \CL |\bxi-\bxi'|\,,\ \ \text{ for all }\,\bxi,\bxi'\in\R^d\,.
\end{equation}
Under this assumption the amplitude $v$ is assumed to satisfy the bounds
\begin{equation}\label{w:eq}
	v(\bxi')\asymp v(\bxi)\,,\ \text{ for all }\,\bxi'\in B\bigl(\bxi, \tau(\bxi)\bigr)\,,
\end{equation}
with implicit constants independent of $\bxi$ and $\bxi'$. It is useful to think of $\tau$ and $v$ as (functional) parameters. They, in turn, may depend on other parameters, e.g. numerical ones like $\a$ and $T$. For example, the results in the previous subsections are based on the assumption that $a$ satisfies \eqref{abounds:eq}, which translates into \eqref{scales:eq} with $\tau(\bxi) = 1$ and $v(\bxi) = \lu\bxi\ru^{-\b}$. On the other hand, in Section \ref{sect:model} we encounter amplitudes and scales depending on the temperature $T$. 

\medskip

\noindent Actually, we will only need the following result involving multi-scale symbols. As mentioned in the Introduction, $\|\,\cdot\,\|_p$ denotes the (quasi-)norm in the Schatten--von Neumann class $\GS_p$ of compact operators. Below the underlying Hilbert space is $\CH = \plainL2(\R^d)$.
  
\begin{prop}\label{prop:cross_smooth} \cite[Lemma 3.4]{LeSpSo3}
Let the set $\L$ satisfy Condition \ref{domain:cond} and let the functions $\tau$ and $v$ be as described above. Suppose that the symbol $a$ satisfies \eqref{scales:eq} and that the conditions
\begin{align}\label{tau_low:eq}
	\a \tau_{\textup{\tiny inf}}\gtrsim 1\,, \quad \tau_{\textup{\tiny inf}} \ceq \inf_{\bxi\in\R^d}\tau(\bxi)>0\,,
\end{align}
hold. Then for any $q\in {]0, 1]}$ we have 
\begin{equation}\label{eq:a_multi}
	\| [\op_\a(a), \chi_\L]\|_q^q\lesssim \a^{d-1} \int_{\R^d} \frac{v(\bxi)^q}{\tau(\bxi)}\,\mathrm{d}\bxi\,.
\end{equation}
This bound is uniform in the symbols $a$ satisfying \eqref{scales:eq} with the same implicit constants. 
\end{prop}

We will make use of \eqref{eq:a_multi} in Section \ref{sect:model} by combining it with bounds obtained in Section \ref{sect:qc}.
 
\section{The positivity of certain entanglement entropies}\label{sect:positive}

Given the organization of the paper, this section is a kind of interlude. It turns out that the property given by its title is present if the underlying entropy function $\eta_\g$, as defined in \eqref{eta1:eq} and \eqref{eq:etaln}, is operator concave (when restricted from the real line $\R$ to its unit interval ${[0,1]}$). Since this and related results are not of asymptotic character, we assume in this section $\a = 1$ for the scaling parameter. 

\subsection{Concavity of the entropy function $\eta_\g$ for $\g\leq 2$}

We prove the property given by the title of this subsection and then establish consequences for the corresponding local and entanglement entropies. The next lemma is elementary. 

\begin{lem} \label{lem:ren}
The entropy function $\eta_\g$ is concave on the interval ${[0, 1]}$ if $\g\in{]0,2]}$, and neither convave nor convex if $\g>2$.
\end{lem}

\begin{proof} 
By the continuity of $\eta_\g$ on ${[0, 1]}$ its enough to check the sign of the second derivative of $\eta_\g$ on the open interval ${]0, 1[}$. For $\g = 1$ we simply have $\eta_1''(t) = - t^{-1}(1-t)^{-1}<0$ so that $\eta_1$ is concave. For $\g\not = 1$ we use the formula
\begin{align}\label{2nd der}
	\eta_\g''(t)[t^\g+(1-t)^\g\big]^2 = -\g[t(1-t)]^{\g-2} - \frac{\g}{1-\g} [t^{\g-1} - (1-t)^{\g-1}]^2\,.
\end{align}
For $\g < 1$ the right-hand side is obviously negative for all $t\in {]0, 1[}$. For $\g =2$ it simply equals $-8t(1-t)<0$. If $\g >2$, then \eqref{2nd der} implies $\eta_\g''(0) = \g/(\g-1)> 0$ and $\eta_\g''(1/2)  = -4\g <0$. Hence $\eta_\g$ is neither concave nor convex. 

It remains to consider the case $\g\in{]1, 2[}$. We rewrite \eqref{2nd der} as   
\begin{align*}
	\eta_\g''(t)[t^\g+(1-t)^\g]^2 = &\  -\frac{\g}{\g-1} g_{\g-1}(t)\,,\\ 
	g_p(t)\ceq &\ p[t(1-t)]^{p-1} - [t^p - (1-t)^p]^2\,,
\end{align*}
for any $p\ceq \g-1\in {]0, 1[}$. Our claim $g_p(t) \ge 0$ is equivalent to 
\begin{align}\label{ineq3}
	[t(1-t)]^{1-p}[t^{2p} + (1-t)^{2p}] \le 2t(1-t) + p\,.
\end{align}
Using the abbreviation
\vspace{-0.6em}
\begin{align*}
	M\ceq 2^{p-1}\max_{t\in[0, 1]} [t^{2p} + (1-t)^{2p}] = 
	\begin{cases}
		2^{-p}& \quad\textup{if}\quad 0<p<1/2\\ 
		2^{p-1}&\quad\textup{if}\quad 1/2\le p<1
	\end{cases}\,,
\end{align*} 
\vspace{-0.6em}
the (elementary example of the) Young inequality
\begin{align*}
	ab\le \frac{a^u}{u} + \frac{b^v}{v}\,,\quad a, b \ge 0\,,\quad u, v> 1\,,\, \frac{1}{u}+\frac{1}{v} = 1
\end{align*}
for $a = [t(1-t)]^{1-p}\,, u = (1-p)^{-1}$ and $b = 1, v = p^{-1}$\, yields   
\begin{align*}
	[(t(1-t)]^{1-p} [t^{2p} + (1-t)^{2p}] \le &\ M [2t(1-t)]^{1-p}\\
	\le &\  M(1-p)[2t(1-t)] + M p\le 2t(1-t) + p\,.
\end{align*}
Since this coincides with \eqref{ineq3}, the proof is complete.  
\end{proof}

The just established concavity is useful to find a lower bound on the local entropy \eqref{def:local E} with $\g\leq 2$, which is larger than the obvious bound $0$. To this end, we recall a formulation \cite[Theorem A.1]{LapSaf1996} of an abstract Jensen-type trace inequality dating back to Berezin \cite{Ber72}. 
\begin{prop} \label{prop:ber}
Let $\CH$ be a complex separable Hilbert space, $P$ an orthogonal projection on $\CH$, $A$ a self-adjoint operator on $\CH$ with its spectrum contained in the interval $I\subset\R$, and $f: I\to\R$ a concave function. Finally, let $\Delta\ceq Pf(PAP)P - Pf(A)P$ be of trace class and PAP compact. Then $\tr\Delta\ge 0$. If $\Delta$ and $PAP$ are of trace class, than also the following trace inequality is valid:
\begin{align*}
	\tr(Pf(PAP)P)\ge\tr(Pf(A)P)\,.
\end{align*}
(If $0\notin I$, then the operator $f(PAP)$ is understood to act on the subspace $P\CH$). 
\end{prop}

The following result is a corollary to Proposition\ref{prop:ber}.
\begin{thm}\label{cor:local} 
Let $\L\subset \R^d$ be bounded and satisfy Condition \ref{domain:cond}. Assume that the Hamiltonian $h$ satisfies \eqref{eq:hin} and that $\g\in {]0, 2]}$. Then the local entropy \eqref{def:local E} obeys the inequality 
\begin{align}\label{eq:local_E}
	\mathrm{S}_{\g}(T,\mu; \L) \ge  s_\g(T, \mu) |\L|\,,
\end{align}
where $s_\g(T, \mu)$ is the entropy density \eqref{density}.
\end{thm}

\begin{proof} 
We use Proposition \ref{prop:ber} with $A = \op_1(a_{T, \mu})$, $P = \chi_{\L}$, and the concave function $f = \eta_\g$. Since  $0\le A\le\mathbbm 1$ and $PAP$ has $\varrho(T,\mu)|\L|$, see \eqref{eq:density}, as its finite trace, Proposition \ref{prop:ber} is indeed applicable and yields $\tr D_1(a, \L; \eta_\g)\ge 0$. By \eqref{eq:local_ED} this entails \eqref{eq:local_E}. 
\end{proof}

We stress that Proposition \ref{prop:ber} cannot be applied if the set  $\L$ is unbounded, since in this case the operator 
$\chi_\L\op_\a(a_{T, \mu})\chi_\L$ is not necessarily compact. Thus Theorem \ref{cor:local} cannot be used to determine the sign of the entanglement entropy \eqref{def:EE}. But, fortunately, we can use the rather strong property as given by the title of the following subsection. 

\subsection{Operator concavity of the entropy function $\eta_\g$ for $\g\leq 1$}

For the general background of this genre we recommend Simon's comprehensive book \cite{Simon2019}. Let $\mathcal H$ be a complex separable Hilbert space of infinite dimension and $\{A, B\}$ an arbitrary pair of bounded self-adjoint operators on $\mathcal H$ with spectra in an  interval $I\subset\R$. A continuous function $f:I \to \R$ is called (decreasing) \emph{operator monotone} if the (operator) inequality $f(A) \ge f(B)$ holds whenever $A\ge B$. Likewise it is called \emph{operator concave} if  $f\big(q A+(1-q)B\big) \ge q f(A) + (1-q)f(B)$ holds for all $q\in {[0,1]}$. It is called \emph{operator convex} if $-f$ is operator concave. Of course, every operator monotone (operator concave) 
function is monotone (concave). We are going to use the following standard examples, see \cite{And}: 
\begin{enumerate}
	\item The function $t\mapsto t^p$, $t\in {[0, \infty[},$ is operator monotone and operator concave if $p\in {]0,1]}$. 
	\item The function $t\mapsto\ln(t)$, $t\in {]0, \infty[}$, is operator monotone and operator concave.  
	\item The function $t\mapsto-t\ln(t)$, $t\in {[0, \infty[}$, is operator concave.  
\end{enumerate} 
Any operator concave function $f$ satisfies \cite{Dav,HP03} the following \emph{Davis operator inequality} for all bounded self-adjoint operators $A$ with spectrum in $I$ and all orthogonal projections $P$ on $\CH$:{}
\begin{align}\label{eq:davis}
	P f\big(PAP\big)P\ge Pf(A)P\,.
\end{align}  
If $0\notin I$, then the operator $f(PAP)$ is understood to act on the subspace $P\CH$. If $0\in I$ and $f(0) = 0$, then \eqref{eq:davis} may be shortened to $f\big(PAP\big)\ge Pf(A)P$.  

\begin{lem} 
If $\g\in {]0,1]} $, then $\eta_\g$ is operator concave on the interval ${[0, 1]}$.
\end{lem}

\begin{proof} 
It suffices to consider self-adjoint operators $A$ and $B$ with $0\le A, B\le \mathbbm 1$. Assume first that $\g < 1$. The function $g_\g(t)\ceq t^\g + (1-t)^\g$ is operator concave on ${[0, 1]}$ (by example 1 above) and the logarithm is operator monotone on ${]0, 1]}$ (example 2). Thus for all $q\in{[0, 1]}$ we have 
\begin{align*}
	\eta_\g\big(q A + (1-q) B\big) & = \frac{1}{1-\g} \, \ln\big[g_\g\big(q A +(1-q)B\big)\big] \\&\ge \frac{1}{1-\g} \ln \big[q g_\g\big(A\big) + (1-q)g_\g\big(B\big)\big]\,. 
\end{align*}
Now, since the logarithm is also operator concave on ${]0,1]}$ (again example 2), the right-hand side is larger than or equal to
\begin{align*}
	\frac{1}{1-\g} \big[q \ln(g_\g(A)) + (1-q) \ln(g_\g(B))\big] =q \eta_\g(A) + (1-q) \eta_\g(B)\,.
\end{align*}
Hence $\eta_\g$, for $\g<1$, is operator concave on ${]0, 1]}$ and, by continuity, on ${[0, 1]}$. For $\g=1$ we proceed more directly and use that $g(t)\ceq - t\ln(t)$ is operator concave on ${[0,1]}$ (example 3).This immediately implies that $\eta_1(t)  = g(t) + g(1-t)$ is also operator concave on ${[0,1]}$.
\end{proof}

Now we can use the inequality \eqref{eq:davis} with $f = \eta_\g$ (for $\g\in {]0,1]}$), $P=\chi_{\L}$, and $A=\op_1(a_{T,\mu})$. Combining this with Proposition \ref{fixedT:prop} yields the following result.
 
\begin{thm} \label{thm:EE positive} 
Let $\L\subset \R^d$ satisfy Condition \ref{domain:cond}. Assume that the Hamiltonian $h$ is as in \eqref{eq:hin} and that $\g\in {]0,1]}$. Then both operators $D_1(a_{T, \mu}, \L; \eta_\g)$ and $D_1(a_{T, \mu}, \R^d\setminus \L; \eta_\g)$ are not only of trace class, but also positive. Hence the entanglement entropy \eqref{def:EE} is positive. 
\end{thm}

This method cannot be used for the EE with $\g>1$ because of the following negative result. 

\begin{thm} 
If $\g>1$, then $\eta_\g$ is not operator concave on the interval $ {[0, 1]} $.
\end{thm}

\begin{proof} For convenience, instead of $\eta_\g$ we consider the function
\begin{align*}
	g_\g(u) \ceq -\eta_{\g}\big(u+\mfr12\big) = \frac{1}{\g-1} \,\ln\big[(\mfr12 -u)^\g + (\mfr12 +u)^\g\big]\,,\quad  u\in \big[-\mfr12, \mfr12\big]\,.
\end{align*}
Our objective is to show that $g_\g$ is not operator convex on $[-1/2, 1/2]$. If $g_\g$ were operator convex, then by \cite[Corollary 1]{BeSh}, $g_\g$ would be analytic on the complex plane with cuts along the half-lines $]-\infty, -1/2[$ and $]1/2, \infty[$. Let us prove that such an analytic continuation of $g_\g$ is impossible. To this end, let $u= \i y/2$ with $y>0$. Then 
\[ 
	\mfr12 \pm u = \mfr12 \sqrt{1+y^2} \,\exp[\pm \i \tan^{-1}(y)]
\]
so that
\begin{align}\label{eq:up}  
	(\mfr12 - u)^\g + (\mfr12 + u)^\g  = 2^{1-\g} (1+y^2)^{\g/2} \, \cos[\g\tan^{-1}(y)]\,.
\end{align}
Since $\g >1$, there exists a finite $y_0>0$  such that $\g\tan^{-1}(y_0) = \pi/2$, so that the right-hand side  of \eqref{eq:up} changes sign at $y = y_0$. This implies that the function $g_\g$ has a branching point at $u = \mathrm{i}y_0/2$, and hence cannot be analytic in the whole upper half-plane. This proves that $g_\g$ is not operator convex, as claimed.
\end{proof}

\section{Quasi-commutator bounds}\label{sect:qc}

In this section we collect some bounds for the Schatten--von Neumann classes $\GS_p$, $p\in {]0,\infty[}$, of compact operators on a complex separable Hilbert space $\CH$, see e.g.~\cite[Chapter 11]{BS}. As mentioned at the end of the Introduction, the functional $\|A\|_p \ceq (\tr(A^*A)^{p/2})^{1/p},\ A\in\GS_p,$  defines a norm for $p\ge 1$ and a quasi-norm for $p<1$. It satisfies the following ``triangle inequality'':
\begin{align}\label{eq:ptriangl}
	\|A+B\|_p^p\le \|A\|_p^p + \|B\|_p^p\,, \quad 0<p\le 1.
\end{align} 
This inequality is used systematically in what follows. The main part is played by estimates for \emph{quasi-commutators} $f(A) J - J f(B)$ with bounded $J$ and bounded self-adjoint $A, B$. The following fact is adapted from \cite[Theorem 2.4]{Sobolev2016}.
 
\begin{prop}\label{prop:perturb}
Suppose that the function $f$ satisfies Condition \ref{cond:f} with some $\d >0$. Let $A, B$ be two bounded self-adjoint operators and let $J$ be a bounded operator. Suppose that $AJ-JB\in\GS_q$ where $q$ satisfies $0 < q < \min\{1, \d\}$. Then
\begin{equation}\label{eq:fest}
	\|f(A) J - J f(B)\|_1 \lesssim \|J\|^{1-q}\big(1+ \|A\|^{\d-q}+\|B\|^{\d-q}\big)\|AJ-JB\|_q^q\,,
\end{equation}
with a constant independent of $A, B$, and $J$. This constant may depend on the set $\CT$  in Condition \ref{cond:f}, and is uniform in the set of functions $f$ satisfying \eqref{eq:fbound} with the same implicit constants.
\end{prop}
Actually, \cite[Theorem 2.4]{Sobolev2016} provides bounds of the type \eqref{eq:fest} in arbitrary (quasi-) normed operator ideals of compact operators and gives a more precise dependence on the constants related to the function $f$. For the present paper \eqref{eq:fest} is sufficient.
  
All subsequent bounds involving the function $f$ are uniform in $f$ in the sense specified in Proposition \ref{prop:perturb}. We are going to apply Proposition \ref{prop:perturb} to obtain various bounds for the operator difference $\CD(A, P; f)\ceq P f(PAP)P - Pf(A)P$ with an orthogonal projection $P$. 

\begin{cor}\label{on:cor}
Let  $f, A, B, J$, and $q$ be as in Proposition \ref{prop:perturb}. Additionally assume that $A, B, J$ satisfy
\begin{align}\label{eq:acomm}
	[A, J] =  [B, J] = 0\,,  \quad (A - B)J = 0\,.
\end{align}
Then
\begin{align}\label{eq:dj}
	\|\CD(A, P; f) J\|_1 + \|J\CD(A, P; f)\|_1 \lesssim   \|[J, P]\|_q^q +  \|[JA, P]\|_q^q\,,{}
\end{align}
and 
\begin{align}\label{eq:on}
	\| \CD(A, P; f)J - J\CD(B, P; f)\|_1\lesssim \|[J, P]\|_q^q +  \|[J, P]\|_1\,.
\end{align}
The implicit constants in these bounds depend on the norms $\|A\|, \|B\|$, and $\|J\|$, but they are uniform in the set of operators $A$, $B$, $J$ whose norms are bounded by the same constants. They are also uniform in the set of functions $f$ in the sense specified in Proposition \ref{prop:perturb}. 
\end{cor}

\begin{proof} 
The proof is based mainly on the bound \eqref{eq:fest}. The assumption \eqref{eq:acomm} considerably simplifies the calculations, and we often use it without mention. 

 For the proof of \eqref{eq:dj} we carry out the estimate for the first term on its left-hand side only, as the second one can be treated in the same way. We begin by writing
\begin{align}\label{eq:fordj}
	\CD(A, P; f) J = &\ P\big( f(PAP)PJ -PJ f(A)\big) - P [f(A), PJ]\,.
\end{align}
Then we use \eqref{eq:fest} and \eqref{eq:ptriangl} to estimate the first term on the right-hand side,
\begin{align*}
	\|Pf(PAP)&PJ - \ PJ f(A)\|_1\le \|f(PAP)PJ - PJ f(A)\|_1\\[0.2cm]
	\lesssim &\   \|PAPJ - PJA\|_q^q 
	\lesssim \   \|P(AJ - JA)P\|_q^q + \|[P, J]\|_q^q  + \|[JA, P]\|_q^q\\
	= &\ \|[P, J]\|_q^q  + \|[JA, P]\|_q^q\,.
\end{align*}
For the second term on the right-hand side of \eqref{eq:fordj} we also use \eqref{eq:fest} and 
\eqref{eq:ptriangl}:
\begin{align*}
	\| P [f(A), PJ]\|_1&\lesssim \  \|APJ - PJA\|_q^q
	\\&\lesssim\  \|(AJ - JA)P\|_q^q + \|[P, J]\|_q^q  + \|[JA, P]\|_q^q\\
	&= \  \|[P, J]\|_q^q  + \|[JA, P]\|_q^q\,.
\end{align*}
Adding up these two estimates we arrive at \eqref{eq:dj}.
 
For the proof of \eqref{eq:on} we first consider the difference  
\begin{align}\label{eq:difir}
	Pf(PA P)PJ - JPf(PBP)P &= \ P\big(f(P A P)J - Jf(P B P)\big)P\notag\\[0.2cm]
	&\ + Pf(P A P)[P, J] - [J, P]f(P BP)P\,. 
\end{align}
We use \eqref{eq:fest} to estimate the first term on the right-hand side as follows
\begin{align*}
	\|P\big(f(P A P)J - Jf(P B P)\big)P\|_1
	\lesssim  \|P A PJ - JP B P\|_q^q
	\le \|[J, P]\|_q^q\,. 
\end{align*} 
To estimate the last two terms on the right-hand side of \eqref{eq:difir}, we notice that $\|f(A)\|\lesssim 1$ and $\|f(B)\|\lesssim 1$ uniformly in $A, B$, and $f$. Consequently,
\begin{align*}
	\|Pf(P A P)[P, J] - [J, P]f(P BP)P\|_1\le  2\|[J, P]\|_1\, \|f\|_{\plainL\infty}\lesssim \|[J, P]\|_1\,.
\end{align*}
%
To summarize, the difference \eqref{eq:difir} has an upper bound like the one in the claim  \eqref{eq:on}.
 
We are going to derive such a bound also for the difference analogous to \eqref{eq:difir}, but  with no $P$ in the argument of $f$: 
\begin{align*}
	Pf(A)PJ - JPf(B)P = &\ P\big(f(A)J - Jf(B)\big)P\\
	& + Pf(A)[P, J] - [J, P]f(B)P\,. 
\end{align*}
In view of \eqref{eq:acomm}, the first term on the right-hand side vanishes. Consequently,  
\begin{align*}
	\|Pf(A)PJ - JPf(B)P\|_1 
	\le 2 \|[P, J]\|_1\, \|f\|_{\plainL\infty}\lesssim \|[J, P]\|_1\,. 
\end{align*}
%
By combining this with the upper bound on \eqref{eq:difir} and the triangle inequality for the trace norm we arrive at \eqref{eq:on}.
\end{proof}
 
\begin{cor} \label{cor:onfull}
Under the assumptions of Corollary \ref{on:cor} (with $\mathbbm 1$ being the identity operator) we have 
\begin{align}\label{eq:onfull}
	\| \CD(A, P; f) - \CD(B, P; f)\|_1 \lesssim\  &\|[J, P]\|_q^q + \ \|[(\mathbbm 1-J)A, P]\|_q^q\notag\\ 
	&\ + \|[(\mathbbm 1-J)B, P]\|_q^q + \|[J, P]\|_1\,.
\end{align}
This bound is uniform in $A, B, J$, and $f$ in the same sense as in Corollary \ref{on:cor}.
\end{cor}
 
\begin{proof}
We observe 
\begin{align*}
	\CD(A, P; f) -  \CD(B, P; f)= &\ \CD(A, P; f)J -  J\CD(B, P; f)\\ 
	&\ + \CD(A, P; f)(\mathbbm 1-J) -  (\mathbbm 1-J)\CD( B, P; f)
\end{align*}
and apply Corollary \ref{on:cor}.
\end{proof}  
 
\section{High-temperature analysis}\label{sect:model}

The purpose of this section is to obtain the large $\a$ and large $T$ asymptotics  for the trace of the operator $D_\a(p_T, \L; f)$ with the symbol $p_T$ of \eqref{eq:pt}, modeling the Fermi symbol \eqref{positiveT:eq} for large $T$. Throughout the section we assume that the function $f$ satisfies Condition \ref{cond:f} with some $\d>0$ and recall that this condition is guaranteed by assumption \eqref{almosthom:eq} as well as by assumption \eqref{homoeta:eq}. We also continue to assume that the truncating region $\L$ satisfies Condition \ref{domain:cond}. Since $\L$ is always fixed, we omit it from the notation and simply write $D_\a(p_T; f)$ and $\CB(p_T; f)$. Recall that $d\ge 2$ throughout. 

\subsection{Further conditions on the single-particle Hamiltonian $h$}\label{subsect:h}

So far we assumed that the Hamiltonian $h\in\plainC\infty(\R^d)$ obeys \eqref{eq:hin} with some $m >0$. More restrictively, from now on we assume that there exists an $m\in\N$ such that
%
\begin{align}\label{eq:h}
	\big|\p_{\bxi}^n h(\bxi)\big|\lesssim \lu\bxi\ru^{2m-|n|}\,,\quad\text{for all }\,n \in\N_0^d\quad\text{and}\,\,\,\,\bxi\in\R^d\,.
\end{align}
Furthermore, we now assume that there exists a function $h_\infty: \R^d\to\R$, being homogeneous of even degree $2m$ (that is, $ h_\infty(t\bxi) = t^{2m} h_\infty(\bxi)$ for all $\bxi\in\R^d$ and all $t >0$), such that
\begin{align}\label{eq:asympt}
	|\bxi|^{-2m}\big|h(\bxi) - h_\infty(\bxi)\big|\to 0\quad \text{ as }\,\, |\bxi|\to\infty\,.
\end{align}
Finally, we reqire that $h_\infty$ is non-degenerate in the sense that
\begin{align}\label{eq:nu}
	2\nu \ceq \min_{|\bxi|=1} h_\infty(\bxi)>0\,.
\end{align} 
Homogeneity and non-degeneracy of $h_\infty$ imply that $h_\infty\ge0$. The conditions \eqref{eq:h}, \eqref{eq:asympt}, and \eqref{eq:nu} also imply that $h$ satisfies \eqref{eq:hin} with the constant $m>0$ from \eqref{eq:h}. We emphasize that from now on this constant is supposed to be integer. This guarantees that $h_\infty\in\plainC\infty(\R^d)$ which enables us to apply the results in Subsection \ref{subsect:asym} to the limiting symbols $1/(1+\mathrm{e}^{h_\infty})$ and $\mathrm{e}^{-h_\infty}$ featured in Section \ref{sect:hta}. 

\subsection{Modeling the Fermi symbol} 
  
Given two positive continuous functions $T\mapsto\phi_T\ge 0$ and $T\mapsto\omega_T>0$ on the temperature half-line ${[1,\infty[}$ with the properties
\begin{align}\label{eq:phiom}
	\phi_T\to \phi_\infty\ge 0\,\quad \text{ and }\quad \om_T\to \om_\infty >0 \quad \textup{ as }\quad T\to \infty\,,
\end{align}
we generalize the Fermi symbol $a_{T,\mu}$ of \eqref{positiveT:eq} to the symbol $p_T$ by the definition
\begin{align}\label{eq:pt}
	p_T(\bxi) \ceq \frac{1}{\phi_T + \om_T \exp\big(h(\bxi)/T)\big)}\,,\quad \bxi\in\R^d\,.
\end{align}
We also consider its ``high-temperature limit'' $p_\infty$ naturally defined by
\begin{align}\label{limit_symbol}
	p_\infty(\bxi) \ceq \frac{1}{\phi_\infty + \om_\infty \exp\big(h_\infty(\bxi)\big)}\,. 
\end{align}

\begin{thm}\label{thm:nonsm} 
Let $p_T$ be the symbol defined in \eqref{eq:pt}. Then 
\begin{align}\label{eq:nonsm}
	\lim \big(\a T^{\frac{1}{2m}}\big)^{1-d} \tr D_\a(p_T; f) =  \CB(p_\infty; f)\,,
\end{align}
as $\a T^{1/2m}\to \infty$ and $T\to\infty$.
\end{thm} 

We also consider the operator $D_\a(\l_T\,p_T; f)$ with the symbol $\l_T\, p_T$, where $\l_T>0$  is, for the time being, an arbitrary continuous function of $T$ that tends to zero as $T\to\infty$. 
  
\begin{thm}\label{thm:sm}
Let $f_0$ be as in Proposition \ref{prop:homo} and $\eta$ be as in \eqref{eta:eq}. Assume that $\a T^{1/2m}\to \infty$ and $T\to\infty$. Then the following implications hold: 
\begin{enumerate}
	\item If $f\in \plainC2(\R\setminus\{0\})\cap \plainC{}(\R)$ satisfies \eqref{almosthom:eq}, then
	\begin{align}\label{eq:smhom}
		\lim \big(\a T^{\frac{1}{2m}}\big)^{1-d}\l_T^{-\d} \tr D_\a(\l_T p_T; f) =  \CB(p_\infty; f_0)\,.
	\end{align}
	\item If $f\in \plainC2(\R\setminus\{0\})\cap \plainC{}(\R)$ satisfies \eqref{homoeta:eq}, then
	\begin{align}\label{eq:smhomlog}
		\lim \big(\a T^{\frac{1}{2m}}\big)^{1-d}\l_T^{-1} \tr D_\a(\l_T\, p_T; f) =  \CB(p_\infty; \eta)\,.
	\end{align}
\end{enumerate}
\end{thm}

To prove these two theorems we compare the operator $D_\a$ for two different symbols defined as follows. Firstly, we pick an arbitrary real-valued ``cut-off'' function $\z\in\plainC\infty(\R^d)$ with $\z(\bxi)=0$ if $|\bxi|\le 1$ and $\z(\bxi)=1$ if $|\bxi|\ge 2$. Moreover, we define two scaled versions of $\z$ by
\begin{align}\label{zeta function}
	\z_T(\bxi) \ceq \z\big(\bxi T^{-\frac{1}{2m}}\big)\,,\quad \widetilde\z_T(\bxi)\ceq\z_T(\bxi/2)\,,\quad  \bxi\in\R^d\,,
\end{align}
so that $\z_T \widetilde\z_T = \widetilde\z_T$. For a fixed number $r \in {]0, 1]}$ we now consider the operators
\begin{align*}
	A=\op_\a(p_T)\,, \quad B=\op_\a(\z_{rT}\, p_T)\,,\quad P=\chi_\L\,, \quad 
	J=\op_\a(\widetilde\z_{rT})\,. 
\end{align*}
They fulfill \eqref{eq:acomm} and their (uniform) norms are uniformly bounded in $T$. Thus we can use Corollary \ref{cor:onfull} for the proof of the following ``comparison lemma'':

\begin{lem}\label{lem:compare}
Assume that $T\gtrsim 1$ and $\a (r T)^{\frac{1}{2m}}\gtrsim 1$ for a fixed $r \in {]0, 1]}$. Then, using \eqref{zeta function}, we have the trace-norm estimate
\begin{align}\label{eq:compare}
	\|D_\a(p_T; f) - D_\a(\z_{rT}\, p_T; f)\|_1\lesssim \a ^{d-1} (rT)^{\frac{d-1}{2m}}\,,
\end{align}
with an implicit constant independent of $\a, T$, and $r$.
\end{lem}

\begin{proof}
Let us estimate the right-hand side of \eqref{eq:onfull} and start with a bound for $\|[J, P]\|_q$, $q\le 1$.  Since $[J, P] = -[\mathbbm 1-J, P]$, we use Proposition \ref{prop:cross_smooth} with $a = 1-\widetilde\z_{rT}$. This symbol satisfies \eqref{scales:eq} with scale and amplitude functions
\begin{align*}
	\tau(\bxi) = (rT)^{\frac{1}{2m}}\,,\quad v(\bxi) = \lu \bxi (rT)^{-\frac{1}{2m}}\ru^{-\b}\,,\quad \bxi\in\R^d\,,
\end{align*}
with an arbitrary $\b>0$. Now we assume that $\b q >d$.  The conditions \eqref{Lip:eq} and \eqref{w:eq} are obviously satisfied, and hence Proposition \ref{prop:cross_smooth} is applicable. We estimate the integral on the right-hand side of \eqref{eq:a_multi}as follows:
\begin{align*}
	\int_{\R^d} \frac{v(\bxi)^q}{\tau(\bxi)} \,
	\mathrm{d}\bxi = (rT)^{-\frac{1}{2m}}\int_{\R^d} \lu \bxi (r T)^{-\frac{1}{2m}}\ru^{-\b q}\,\mathrm{d}\bxi
	\lesssim (rT)^{\frac{d-1}{2m}}\,.  
\end{align*}
Thus, under our assumptions on $\a, T$, and $r$ the condition \eqref{tau_low:eq} is satisfied, and hence, by \eqref{eq:a_multi}, we have
\begin{align*}
	\|[J, P]\|_q^q  = \|[\op_\a(a), \chi_\L]\|_q^q\lesssim \a^{d-1}(rT)^{\frac{d-1}{2m}}\,.  
\end{align*}
Estimating $\|[(\mathbbm 1-J)A, P]\|_q$ and $\|[(\mathbbm 1-J)B, P]\|_q$ is somewhat trickier. For the first commutator we are going to use Proposition \ref{prop:cross_smooth} with the symbol  $a = (1-\widetilde\z_{rT})p_T$. At first we estimate the derivatives of ${\mathrm e}^{h(\bxi)/T}$ for $|\bxi|\le 4(rT)^{\frac{1}{2m}}$ using \eqref{eq:h}:
\begin{align*}
	\big|\p_{\bxi}^k \,{\mathrm e}^{\frac{h(\bxi)}{T}}\big|\lesssim 
	\lu\bxi\ru^{-|k|}{\mathrm e}^{\frac{h(\bxi)}{T}}
	\lesssim \lu\bxi\ru^{-|k|}\,,\quad k\in\N_0^d\,.
\end{align*} 
\newpage
\noindent Furthermore,
\begin{align*}
	\big|\p_{\bxi}^k \big(1-\widetilde\z_{rT}(\bxi)\big)\big|\lesssim (rT)^{-\frac{|k|}{2m}}
	\chi_{\{|\bxi|\le 4(rT)^{\frac{1}{2m}}\}}(\bxi)\lesssim  \lu\xi\ru^{-|k|}\,
	\chi_{\{|\bxi|\le 4(rT)^{\frac{1}{2m}}\}}(\bxi)\,.
\end{align*}
Therefore, we obtain from \eqref{eq:pt} that 
\begin{align*}
	\big|\p_{\bxi}^k \,a(\bxi)\big|\lesssim \lu \bxi\ru^{-|k|}\lu \bxi (rT)^{-\frac{1}{2m}}\ru^{-\b}\,.
\end{align*}
with an arbitrary $\b >d/q$. Consequently, the symbol $a$ satisfies  \eqref{scales:eq} with the scale and amplitude 
\begin{align*}
	\tau(\bxi) = \frac{1}{2}\lu\bxi\ru\,,\quad 
	v(\bxi) = \lu \bxi (rT)^{-\frac{1}{2m}}\ru^{-\b}\,,\quad \bxi\in\R^d\,.
\end{align*} 
Again the conditions \eqref{Lip:eq}, \eqref{w:eq}, and \eqref{tau_low:eq} are satisfied, and we can use Proposition \ref{prop:cross_smooth} to produce the bound
\begin{align*}
	\int_{\R^d} \frac{v(\bxi)^q}{\tau(\bxi)}\, \mathrm{d}\bxi\le 2\int_{\R^d} |\bxi|^{-1}
	\lu \bxi (rT)^{-\frac{1}{2m}}\ru^{-\b q} \, \mathrm{d}\bxi
	\lesssim (rT)^{\frac{d-1}{2m}}\,.
\end{align*}  
Thus by \eqref{eq:a_multi}, 
\begin{align*}
	\|[(\mathbbm 1-J)A, P]\|_q^q = \|[\op_\a(a), \chi_\L]\|_q^q\lesssim \a^{d-1} (rT)^{\frac{d-1}{2m}}\,.  
\end{align*}
The bound for the commutator $[(\mathbbm 1-J) B, P]$ is proved in the same way.  Substituting the above bounds into the statement \eqref{eq:onfull} of Corollary \ref{cor:onfull}, we get the claimed estimate \eqref{eq:compare}.
\end{proof}

\noindent A useful consequence of this fact is the following continuity statement: 

\begin{cor}\label{cor:contbn} With the function $\z_r$ defined in \eqref{zeta function} and the symbol $p_\infty$ defined in \eqref{limit_symbol} we have
\begin{align}\label{eq:zr}
	\lim_{r\to 0} \CB(\z_r p_\infty; f) = \CB(p_\infty; f)\,.
\end{align}
\end{cor}

\begin{proof}
We apply Lemma \ref{lem:compare} with $h=h_\infty$, $T=1$, and the constant functions $\om \equiv \om_\infty$ and $\phi \equiv \phi_\infty$. Then, for $\a r\gtrsim 1$,
\begin{align*}
	\|D_\a(p_\infty; f) - D_\a(\z_{r} p_\infty; f)\|_1\lesssim \a ^{d-1} r^{\frac{d-1}{2m}}\,.
\end{align*}
Therefore,
\begin{align*}
	\big|\a^{1-d} \tr D_\a(p_\infty; f) - \a^{1-d} \tr D_\a(\z_r p_\infty; f)\big|\le r^{\frac{d-1}{2m}}\,.
\end{align*}
Now, we use Proposition \ref{fixedT:prop} to obtain the estimate 
\begin{align*}
	\big|\CB(p_\infty; f) - \CB(\z_r p_\infty; f)\big|\lesssim r^{\frac{d-1}{2m}}\,.
\end{align*}
This leads to \eqref{eq:zr}, as claimed. 
\end{proof}

We already have a result on the continuity of the asymptotic coefficient, see Proposition \ref{prop:contb}. However, this proposition is not applicable to the truncated symbol $\z_r\,p_T$, since its derivatives are not bounded uniformly in $r>0$. This explains why we need Corollary \ref{cor:contbn}.

The next lemma provides the same asymptotics as in Theorems \ref{thm:nonsm} and \ref{thm:sm}, but this time for $ \z_{rT}\, p_T $ instead of $p_T$. We recall that $\l_T>0$ obeys $\l_T\to 0$ as $T\to\infty$.

\begin{lem}
Assume that $r \in {]0, 1]}$ is fixed and that $\a T^{\frac{1}{2m}}\to\infty, T\to\infty$. Then  
\begin{align} \label{eq:cutoff}
	\lim \big(\a T^{\frac{1}{2m}}\big)^{1-d}\tr D_\a(\z_{rT}\, p_T; f) =  \CB(\z_r\, p_\infty; f)\,.
\end{align}
If $f$ satisfies \eqref{almosthom:eq}, then
\begin{align}\label{eq:cutoffd}
	\lim \big(\a T^{\frac{1}{2m}}\big)^{1-d}\l_T^{-\d} \tr D_\a(\l_T \z_{rT}\, p_T; f) =  \CB(\z_r\, p_\infty; f_0)\,.
\end{align}
If $f$ satisfies \eqref{homoeta:eq}, then 
\begin{align}\label{eq:cutofflog}
	\lim \big(\a T^{\frac{1}{2m}}\big)^{1-d}\l_T^{-1} \tr D_\a(\l_T \z_{rT}\, p_T; f) =  \CB(\z_r\, p_\infty; \eta)\,.
\end{align}
\end{lem}

\begin{proof}
By a straightforward change of variables in the definition \eqref{eq:opa}, we obtain
\begin{align*}
	\op_\a(\z_{rT}\,p_T) = \op_L(b_T)\quad \textup{and} \quad D_\a(\z_{rT}\, p_T; f) = 	D_L(b_T; f)\,, 
\end{align*}
where
\begin{align*}
L \ceq \a T^{\frac{1}{2m}}\quad \textup{and}\quad b_T(\bxi) \ceq \z_r(\bxi) p_T(T^{\frac{1}{2m}} \bxi)\,,\quad\bxi\in\R^d\,.
\end{align*}
Thanks to the condition \eqref{eq:asympt}, for all $\bxi\not = \bf{0}$, we have as $T\to\infty$:
\begin{align*}
	T^{-1}h(T^{\frac{1}{2m}}\bxi) \to h_\infty(\bxi)\quad \textup{and hence by 		\eqref{eq:phiom}} \quad   
	b_{T}(\bxi) \to \z_r(\bxi) p_\infty(\bxi)\,.
\end{align*}
Assuming that $T\gtrsim 1$, an elementary calculation using \eqref{eq:h} leads to the bounds 
\begin{align}\label{atmu:eq}
	\big|\p_{\bxi}^n b_{T}(\bxi)\big| + \big|\p_{\bxi}^n \big(\z_r(\bxi)p_\infty(\bxi)\big)\big| 
	\lesssim \mathrm{e}^{-\nu |\bxi|^{2m}}\,, \ n\in\N_0^d\,,\,\bxi\in\R^d\,, 
\end{align}
with $\nu>0$ from \eqref{eq:nu} and implicit constants depending on the number $r \in {]0, 1]}$. Since $b_T$ satisfies \eqref{atmu:eq} uniformly in $T\gtrsim 1$, we obtain by Corollary \ref{prop:uni} that 
\begin{align*}
	\lim\limits_{L\to\infty, T\to\infty}L^{1-d} D_L(b_T; f) = \CB(\z_r p_\infty; f)\,. 
\end{align*}
By the above change of variables, this leads to \eqref{eq:cutoff}. Formulas \eqref{eq:cutoffd} and \eqref{eq:cutofflog} follow along the same lines from Propositions \ref{prop:homo} and \ref{prop:homolog}. 
\end{proof}
\newpage
\begin{proof}[of Theorem \ref{thm:nonsm}] 
We begin by estimating as follows:
\begin{align*}
	\big|\a^{-(d-1)} &T^{- \frac{d-1}{2m}} \tr D_\a(p_T; f) - \CB(p_\infty; f)\big|
	\\\le & \ \big(\a T^{\frac{1}{2m}}\big)^{1-d}\|D_\a(p_T; f) - D_\a(\z_{rT}\, p_T; f)\|_1\\ 
	& + \big|\big(\a T^{\frac{1}{2m}}\big)^{1-d}\tr D_\a(\z_{rT}\, p_T; f) - \CB(\z_r\, p_\infty; f)\big|\\
	& + |\CB(\z_r\, p_\infty; f) - \CB(p_\infty; f)\big|\,.
\end{align*}
By \eqref{eq:compare} and \eqref{eq:cutoff} we then obtain 
\begin{align*}
	\limsup\big|\a^{-(d-1)} T^{- \frac{d-1}{2m}} &\ \tr D_\a(p_T; f) - \CB(p_\infty; f)\big|\\
	&\le r^{\frac{d-1}{2m}} + |\CB(\z_r\, p_\infty; f) - \CB(p_\infty; f)\big|\,,
\end{align*}
where the upper limit is taken as $ \a T^{\frac{1}{2m}}\to\infty$, $T\to\infty$.  Taking $r\to 0$ and using \eqref{eq:zr} we arrive at \eqref{eq:nonsm}.
\end{proof}

\begin{proof}[of Theorem \ref{thm:sm}] 
We recall that the only singular point of the function $f$ is $t = 0$. We assume that $f$ satisfies \eqref{almosthom:eq}, so that for all $t\not = 0$,
\begin{align}\label{eq:sing}
|f^{(k)}(t)|\lesssim |t|^{\d - k}, \quad k\in\{ 0, 1, 2 \}\,.
\end{align}
Consequently, the function $\tilde{f}_T(t) \ceq \l_T^{-\d}\, f(\l_T\, t)$, $t\in\R$, satisfies the same inequalities with the same constants. Now we can apply the argument in the previous proof to the operator  
\begin{align*}
	D_\a(p_T; \tilde{f}_T) = \l_T^{-\d} D_\a(\l_T p_T; f)\,.
\end{align*}
More precisely, we estimate as follows
\begin{align}\label{eq:chain}
	\big|\a^{-(d-1)} &T^{- \frac{d-1}{2m}} \l_T^{-\d} \ \tr D_\a(\l_T\, p_T; f) - \CB(p_\infty; f_0)\big|\notag\\
	\le &\ \big(\a T^{\frac{1}{2m}}\big)^{1-d}
	\big\|D_\a(p_T; \tilde{f}_T) - D_\a(\z_{rT}\, p_T; \tilde{f}_T)\big\|_1\notag\\ 
	&\  + \big|\big(\a T^{\frac{1}{2m}}\big)^{1-d} \l_T^{-\d}
	\tr D_\a(\l_T \z_{rT}\, p_T; f) - \CB(\z_r\, p_\infty; f_0)\big|\notag\\
	&\  + |\CB(\z_r\, p_\infty; f_0) - \CB(p_\infty; f_0)\big|\,.
\end{align}
By \eqref{eq:compare} and \eqref{eq:cutoffd} we then obtain 
\begin{align*}
	\limsup\big|\a^{-(d-1)} T^{- \frac{d-1}{2m}} \l_T^{-\d} 
	&\ \tr D_\a(\l_T p_T; f) - \CB(p_\infty; f_0)\big|\\
	\le &\ r^{\frac{d-1}{2m}} + |\CB(\z_r\, p_\infty; f_0) - \CB(p_\infty; f_0)\big|\,,
\end{align*}
where the upper limit is again taken as $ \a T^{\frac{1}{2m}}\to\infty$, $T\to\infty$. Taking $r\to 0$ and using \eqref{eq:zr} we obtain \eqref{eq:smhom}.

Now we assume that $f$ satisfies \eqref{homoeta:eq}. We use \eqref{eq:chain} with $\d = 1$ and $f_0$ replaced by $\eta$. Then
\begin{align*}
	\big|\a^{-(d-1)} T^{- \frac{d-1}{2m}} \l_T^{-1} &\ 
	\tr D_\a(\l_T\, p_T; f) - \CB(p_\infty; \eta)\big|\notag\\
	\le &\ \big(\a T^{\frac{1}{2m}}\big)^{1-d}
	\big\|D_\a(p_T; \tilde{f}_T) - D_\a(\z_{rT}\, p_T; \tilde{f}_T)\big\|_1\notag\\ 
	&\ \quad + \big|\big(\a T^{\frac{1}{2m}}\big)^{1-d} \l_T^{-1}
	\tr D_\a(\l_T\, \z_{rT}\, p_T; f) - \CB(\z_r\, p_\infty; \eta)\big|\notag\\
	&\ \quad + |\CB(\z_r\, p_\infty; \eta) - \CB(p_\infty; \eta)\big|\,.
\end{align*}
As before, the last term on the right-hand side tends to zero due to \eqref{eq:zr}. The second term vanishes as $ \a T^{\frac{1}{2m}}\to\infty$, $T\to\infty$ due to \eqref{eq:cutofflog}. To estimate the first term we observe that $g\ceq f-\eta $ satisfies \eqref{eq:sing} with $\d = 1$. Therefore,
\begin{align*}
	D_\a(p_T; \tilde{f}_T) - D_\a(\z_{rT}\, p_T; \tilde{f}_T) = &\ \big[D_\a(p_T; \tilde\eta_T) - D_\a(\z_{rT}\, p_T; \tilde\eta_T)\big]\\
	&\ \quad + \big[D_\a(p_T; \tilde g_T) - D_\a(\z_{rT}\, p_T; \tilde g_T)\big]\,,
\end{align*}
where $\tilde{\eta}_T(t)\ceq \l_T^{-1} \eta(\l_T\, t)$ and $\tilde{g}_T(t)\ceq \l_T^{-1} g(\l_T\, t)$. As in the previous calculation, the second term is estimated with the help of \eqref{eq:compare} by $r^{\frac{d-1}{2m}}$. Since $\tilde\eta_T(t) = \eta(t) - t\ln(\l_T)$ and the operator difference \eqref{Dalpha:eq} vanishes for linear functions $f$, we have  
\begin{align*}
	D_\a(p_T; \tilde\eta_T) - D_\a(\z_{rT}\, p_T; \tilde\eta_T) = D_\a(p_T; \eta) - D_\a(\z_{rT}\, p_T; \eta)\,.
\end{align*}
The function $\eta$ from \eqref{eta:eq} satisfies \eqref{eq:sing} for all $\d <1$, and hence by \eqref{eq:compare} the above difference is again estimated by $r^{\frac{d-1}{2m}}$. This entails \eqref{eq:smhomlog} and the proof of Theorem \ref{thm:sm} is complete. 
\end{proof}

\section{Main results on the high-temperature asymptotics}\label{sect:hta}

In this section we adapt Theorems \ref{thm:nonsm} and \ref{thm:sm} to two different asymptotic regimes of the entanglement entropy \eqref{def:EE}, when the temperature becomes large. This is straightforward for the (first) regime of a fixed chemical potential $\mu$, since we work from the outset within the grand-canonical formalism \cite{Balian92,Bratteli97} for an indefinite number of particles. For the (second) regime of a fixed particle density $\rho$ it is slightly more involved, but physically often more interesting. Both results will be discussed in some detail in Subsection \ref{subsect:hta}. 

\subsection{Case of a fixed chemical potential $\mu$} 

Since the Fermi symbol $a_{T, \mu}$ of \eqref{positiveT:eq} equals $p_T$ of \eqref{eq:pt} with $\phi_T=1$ and $\om_T =\exp(-\mu/T)$, the limiting symbol in this case is obviously $p_\infty= 1/(1+\mathrm{e}^{h_\infty})$. For the function $f$ we take $\eta_\g$  which satisfies Condition \ref{cond:f} with $\CT= \{0, 1\}$ and an arbitrary $\d < \min\{1, \g\}$. Thus, by combining the definition \eqref{def:EE}, Remark \ref{rem:signb}(1), and Theorem \ref{thm:nonsm} we obtain: 

\begin{thm}\label{thm:fixedmu}
Let the truncating region $\L$ satisfy Condition \ref{domain:cond}. Then we have
\begin{align}\label{eq:fixedmuasymp}
	\lim\big(\a T^{\textstyle\frac{1}{2m}}\big)^{1-d}\,\mathrm{H}_{\g}(T,\mu;\a\L) = 2\, \CB\big((1+\mathrm{e}^{h_\infty})^{-1},\p\L; \eta_\g\big)
\end{align}
for any fixed $\mu\in\R$, as $\a T^{\frac{1}{2m}}\to\infty$ and $T\to\infty$.
\end{thm}

\vspace{-2em}
\subsection{Case of a fixed particle density $\rho$}
\vspace{-1em}
In this case we have to find a function $T\mapsto \mu_\rho(T)$ satisfying \eqref{rhofixed} for a fixed constant $\rho>0$. According to the Appendix we have for any such function
\begin{align}\label{rho}
	\exp\big(-\mu_\rho(T)/T\big)= \lambda_T^{-1}\big(1+o(1)\big) \quad \text{ as } \quad T\to\infty\,,
\end{align}
with the function $T\mapsto \lambda_T$ given explicitly by
\begin{align}\label{lambda}
	\lambda_T\ceq\rho T^{-\frac{d}{2m}}/\varkappa \quad\text{where}\quad\varkappa\ceq(2\pi)^{-d}\int_{\R^d} \mathrm{\exp}(-h_\infty(\bxi))\, \mathrm{d}\bxi\,.
\end{align}
This implies for the Fermi symbol at fixed $\rho$ the formula
\begin{align}\label{a1}
	a_{T, \mu_\rho(T)}=  \lambda_T p_T\,, 
\end{align}
where $p_T$ is given by \eqref{eq:pt} with
\begin{align}\label{a2}
	\phi_T = \l_T\quad \text{ and }\quad  \om_T = \l_T  \exp\big(-\mu_\rho(T)/T\big)\,.
\end{align}
By \eqref{rho} and \eqref{lambda} we clearly have $\phi_T\to0$ and $\om_T\to 1$ as $T\to\infty$. Hence the limiting symbol is the classical ``Boltzmann factor'' corresponding to $h_\infty$: 
\begin{align}\label{eq:F0}
	p_\infty(\bxi) = \mathrm{e}^{-h_\infty(\bxi)}\,,\quad \bxi\in\R^d\,.
\end{align} 
To study the symbol $a_{T, \mu_\rho(T)}$ we use Theorem \ref{thm:sm} with $f = \eta_\g$. For a start we need to understand the behavior of $\eta_\g(t)$ for small $t$. Since this behavior depends on $\g$, we have to distinguish five different regimes for its value. To state the result in a unified way for all values, we define a parameter $\d_\g>0$ and a pair of functions \{$f_\g$, $\eta_\g^{\mathrm{eff}}$\} on the interval ${]0,1[}$ in Table \ref{table:asymp}. 

The next lemma shows that $\eta_\g^{\mathrm{eff}}$ describes the effective asymptotic contribution of $f_\g$ as $t\downarrow 0$.
\begin{lem} 
Let the parameter $\d_\g$ and the functions $f_\g$, $\eta_\g^{\mathrm{eff}}$ on the interval ${]0, 1[}$ be as defined in Table \ref{table:asymp}. Then the following three asymptotic relations hold:
\begin{align}\label{eq:homgam}
	\lim\limits_{t\to 0} t^{k-\d_\g}\frac{\mathrm{d}^k}{\mathrm{d}t^k} 
	\big( f_\g(t) - \eta_\g^{\mathrm{eff}}(t)\big) = 0\,,\quad k\in \{0, 1, 2\}\,.
\end{align}   
\end{lem}
\begin{proof}

For $\g\not = 1$ we expand $\eta_\g$ at $t=0$ to obtain
\begin{align}\label{eq:logtay}
	(1-\g)\eta_\g(t) = t^\g -\g t - \frac{\g}{2} t^2 + O(t^3) + O(t^{2\g}) + O(t^{1+\g})\,. 
\end{align}
\setlength{\tabcolsep}{10pt} 
\renewcommand{\arraystretch}{1.5}
\begin{table}[h!]
\centering
\begin{tabular}{|c||c|c|c|} 
 \hline
   & $\d_\g$ & $f_\g(t)$ & $\eta_\g^{\mathrm{eff}}(t)$\\ [0.2cm]
 \hline\hline
 $0<\g<1$ & $\g$ & $\eta_\g(t)$ & $\frac{1}{1-\g} t^\g $ \\[0.2cm] 
 \hline
 $\g=1$ & $1$ & $\eta_1(t) - t$ & $-t\ln(t)$\\[0.2cm]
 \hline
 $1<\g<2$ & $\g$ & $\eta_\g(t)-\frac{\g}{\g-1}t$ & $\frac{1}{1-\g} t^\g $\\[0.2cm]
 \hline
 $\g=2$  & $3$ & $\eta_2(t)-2t$ & $-\frac{4}{3}t^3$ \\[0.2cm]
 \hline
 $\g >2$ & $2$ & $\eta_\g(t)-\frac{\g}{\g-1}t$ & $\frac{\g}{2(\g-1)}t^2$ \\ [0.2cm] 
 \hline  
\end{tabular}

\caption{The parameter $\d_\g$ and the functions \{$f_\g$, $\eta_\g^{\mathrm{eff}}$\} for different values of the R\'enyi index $\g$}\label{table:asymp}
\end{table}

The notation $g(t) = O(t^b)$ means that $|g^{(k)}(t)|\lesssim t^{b-k}$, for all $k\in\N_0$. This expansion leads to the claim \eqref{eq:homgam} for all $\g\notin\{1,2\}$. For $\g =2$ the expansion \eqref{eq:logtay} is insufficient since the terms with $t^2$ cancel out. By including the third-order term explicitly we find $\eta_2(t) = 2t - \frac{4}{3} t^3 + O(t^4)$ and obtain \eqref{eq:homgam} for $\g =2$. Finally, for $\g=1$ we have $\eta_1(t) =-t\ln(|t|) + t + O(t^2)$ which again leads to \eqref{eq:homgam}. 
\end{proof}

Now we are in a position to state and prove our second main result for the scaling of the entanglement entropy. 

\begin{thm}\label{thm:fixedrho} 
Let the truncating region $\L$ satisfy Condition \ref{domain:cond}. For a number $\rho>0$ let $T\mapsto\mu_\rho(T)$ be a function satisfying \eqref{rhofixed} and $\varkappa$ be as defined in \eqref{lambda}. Finally, let $\d_\g$  and $\eta_\g^{\mathrm{eff}}$ be as defined in Table \ref{table:asymp}. Then we have the asymptotic relation
\begin{align}\label{eq:fixedrhoasymp}
\lim \big(\a T^{\textstyle\frac{1}{2m}}\big)^{1-d}\,\big(\varkappa \,T^{\textstyle\frac{d}{2m}}/\rho\big)^{\d_\g}\,\mathrm{H}_{\g}\big(T,\mu_\rho(T);\a\L\big) 
	= 2\, \CB\big( \mathrm{e}^{-h_\infty}, \p\L; \eta_\g^{\mathrm{eff}}\big)\,
\end{align}
for any fixed $\rho>0$, as\, $\a T^{\frac{1}{2m}}\to \infty\,\, \textup{and}\,\,T\to\infty$\,. 
\end{thm} 
 
\begin{proof} 
We replace the symbol $a_{T, \mu_\rho(T)}$ with $\l_T\, p_T$ as specified in \eqref{a2}. Furthermore, since the operator difference \eqref{Dalpha:eq} vanishes for linear functions $f$, we have
\begin{align*}
	D_\a(\l_T\, p_T, \L; \eta_\g) = D_\a(\l_T\, p_T, \L; f_\g)\,. 
\end{align*}
As the function $f_\g$ satisfies the condition \eqref{eq:homgam}, we can use Theorem \ref{thm:sm} which gives 
\begin{align*}
	\lim \big(\a T^{\frac{1}{2m}}\big)^{1-d} \l_T^{-\d_\g} \tr D_\a(\l_T\, p_T, \L; f_\g) = \CB(p_\infty, \p\L; \eta_\g^{\mathrm{eff}})\,,
\end{align*}
as $\a T^{\frac{1}{2m}}\to \infty$ and $T\to\infty$, where $p_\infty$ is given by \eqref{eq:F0}. The same formula holds for the region $\R^d\setminus\overline\L$. The claimed formula \eqref{eq:fixedrhoasymp} now follows from the definition \eqref{def:EE}, Remark \ref{rem:signb}(1), \eqref{lambda}, and \eqref{a1}.
\end{proof}
  
\subsection{Concluding remarks}\label{subsect:hta} 

\begin{enumerate}
	\item We have proved the positivity of the EE for R\'enyi index $\g\le1$, see Theorem \ref{thm:EE positive}. For bounded $\L$ we have actually the stronger statement:
	\[
		0\le \mathrm{S}_{\g}(T,\mu; \L)-s_\g(T, \mu) |\L|\le \mathrm{H}_{\g}(T,\mu;\L)\,, 			\quad \g\le 1\,.{}
	\]
	Here the first inequality for the local entropy is due to Theorem \ref{eq:local_E} and holds even for $\g\le2$. The second one follows by combining \eqref{def:EE}, \eqref{eq:local_ED}, and Theorem \ref{thm:EE positive}. Although Theorem \ref{thm:EE positive} ensures the positivity of the EE only for $\g\le 1$, its asymptotic coefficient $\CB(a_{T, \mu},\p\L; \eta_\g)$ is positive for all $\g\le2$ (and all $\mu\in\R$, $T>0$), see Remark \ref{rem:signb}(2). It is an open question whether the positivity of the EE itself extends to all $\g\le 2$.   
	\item \label{item:negative}
	In case of a fixed $\rho$ (Theorem \ref{thm:fixedrho}) and $\g\le2$, the function $\eta_\g^{\mathrm{eff}}$ in Table \ref{table:asymp} is strictly concave so that, by $h_\infty\geq 0$ and Remark \ref{rem:signb}, the coefficient $\CB\big( \mathrm{e}^{-h_\infty}, \p\L; \eta_\g^{\mathrm{eff}}\big)$ is strictly positive. On the other hand, if $\g>2$, then $\eta_\g^{\mathrm{eff}}$ is strictly convex so that this coefficient is strictly \emph{negative}. This change of sign, unexpected by us, suggests that our definition \eqref{def:EE} of the EE is somewhat unphysical for $\g>2$. 
	In passing we note that $\CB\big( \mathrm{e}^{-h_\infty}, \p\L; \eta_\g^{\mathrm{eff}}\big)$ for $\g >2$ can be computed rather explicitly. Indeed, by Table \ref{table:asymp} the
	integral \eqref{U:eq} becomes simply:
	\[
		U\big(u,v;\eta_\gamma^{\mathrm{eff}}\big) = -\frac{\g}{2(\g-1)} (u-v)^2\,,\quad \g>2\,.{}
	\]
	This observation and the use of Parseval's identity, as in the proof of \cite[Proposition 1]{Widom1982}, enables us to determine the resulting quantities \eqref{ad:eq} and hence \eqref{cbd:eq} in terms of the Fourier transform of $\mathrm{e}^{-h_\infty} $. 
	In particular, for the quadratic ``asymptotic'' Hamiltonian $h_\infty(\bxi) = |\bxi|^2/2$ (which includes the Hamiltonian of the ideal Fermi gas) we obtain the formula 
	\[
		\CB\big( \mathrm{e}^{-h_\infty}, \p\L; \eta_\g^{\mathrm{eff}}\big) = -  \frac{1}{4}(4\pi)^{-(d+1)/2}\frac{\g}{\g-1}|\p\L|\,,\quad \g>2\,.
	\]

	\item For $\g<2$ the function $\eta_\g^{\mathrm{eff}}$ is exactly the classical entropy function of the Maxwell--Boltzmann gas. This confirms our expectations stated in the Introduction. This conclusion does not hold for $\g\ge 2$. 
 
	\item In case of a fixed $\mu$ (Theorem \ref{thm:fixedmu}) and $\g\le2$ we know that the coefficient $\CB_\g: = \mathcal B(p_\infty,\p\L;\eta_\g)$, with 
	$p_\infty = 1/(1+\mathrm{e}^{ h_\infty})$, on the right-hand side of \eqref{eq:fixedmuasymp} is positive by the strict concavity of $\eta_\g$. Here we want to indicate that $\CB_\g$ should become negative for large $\g$. To this end, we consider the point-wise limit:
	\[
		\eta_\infty(t)\coloneqq \lim_{\g\to\infty}\eta_\g(t) = \min\big\{-\ln(1-t),-\ln(t)\big\}\,,\quad t\in [0,1]\,.
	\]
	This function is (strictly) convex on ${[0,1/2]}$ and also on ${[1/2,1]}$, but not on the whole interval ${[0,1]}$. Fortunately however, since $p_\infty\le 1/2$  by our assumptions in Subsection \ref{subsect:h}, the integral $U \big(u,v;\eta_\infty\big)$ from \eqref{U:eq} underlying $\mathcal B_\infty$ enters only for $u,v\in {[0,1/2]}$ so that	
	\[
		\mathcal B_\infty = \mathcal B\big(p_\infty, \p\L; -\ln(1-\cdot)\big)<0\,.
	\]
To get more specific about the value of $\mathcal B_\infty$ we observe that for $u,v \in [0,1/2]$:
	\begin{align*} 
		U \big(u,v;\eta_\infty\big) &= - U\big(u,v;\ln(1-\cdot)\big) = - U\big(1-u,1-v;\ln(\cdot)\big) \\
		&= -\frac{1}{2} \big(\ln(1-u) - \ln(1-v)\big)^2\,,
	\end{align*}
	where the last equality is an elementary calculation, see again \cite{Widom1982}. It reconfirms the negativity of $U$. Introducing the symbol $q_\infty \coloneqq \ln(1+\mathrm{e}^{-h_\infty})$ we get
	\begin{equation*} 
		U\big(p_\infty(\bxi), p_\infty(\bxi + t\mathbf{e}); \eta_\infty\big) = -\frac{1}{2}\big(q_\infty(\bxi) - q_\infty(\bxi+t\mathbf{e})\big)^2\,.
	\end{equation*}
	Similarly to Remark \ref{item:negative}, the coefficient $\CB_\infty$ can now be found in terms of the Fourier transform of (the Taylor series of) $q_\infty$. In particular, for $h_\infty(\bxi) = |\bxi|^2/2$ the coefficient $\CB_\infty$ can be computed explicitly:
	\[ 
		\mathcal B\big((1+\mathrm{e}^ { h_\infty})^{-1},\p\L;\eta_\infty\big) = - \frac{1}{2}(2\pi)^{-(d+1)/2} \Sigma(d) |\p\L|\,,
	\]
	where $\Sigma(d)\coloneqq \sum_{n,m\ge1} (-1)^{n+m} (nm)^{-1/2} (n+m)^{-(d+1)/2}$. We omit the details. [Numerically we find e.g. $\Sigma(2) \approx 0.19798$ and $\Sigma(3) \approx 0.15419$.]
Thus, assuming continuity of $\mathcal B_\g$ as a function of $\g>0$, we can claim the existence of a finite value $\g_0>2$ such that $\mathcal B_{\g_0} = 0$ and $\mathcal B_\g < 0$ for certain finite $\g>\g_0$.

	\item\label{item:afixed} In both high-temperature regimes (fixed $\mu$ or fixed $\rho$) the scaling parameter $\a$ may be fixed, to $\a=1$, say. Then the truncating set $\a\L$ is fixed and only the temperature $T$ tends to infinity.

	\item Using the relation \eqref{eq:local_ED} and Theorem \ref{thm:nonsm} one can also derive appropriate asymptotic formulas for the local entropy. For example, assuming that $\mu$ is fixed, as in Theorem \ref{thm:fixedmu}, we easily obtain the asymptotic relation
	\begin{align*}
		\lim\big(\a T^{\textstyle\frac{1}{2m}}\big)^{1-d}\,\big[\mathrm{S}_{\g}(T,\mu;\a\L) - \a^d  s_\g(T, \mu)\,|\L|\big]
	= \CB\big((1+\mathrm{e}^{h_\infty})^{-1}, \p\L; \eta_\g\big)\,,
	\end{align*}
	as $\a T^{\frac{1}{2m}}\,\to \infty$ and $T\to\infty$. However, in order to obtain from this formula a proper asymptotic expansion for the local entropy, one needs to find an expansion for the entropy density $s_\g(T, \mu)$ as $T\to\infty$. A convenient starting point for its derivation could be, for example, the representation (10.8) in \cite{LeSpSo3}. 

	An analogous formula can be written down for the case of a fixed $\rho$. The inequality $\CB\big(\mathrm{e}^{-h_\infty}, \p\L;\eta_\g^{\mathrm{eff}}\big)<0$ for $\g >2$ would then imply that the local entropy obeys for large $T$ the bound
	${\mathrm S}_\g\big(T, \mu_\rho(T); \L\big) < s_\g\big(T, \mu_\rho(T)\big) |\L|$, instead of \eqref{eq:local_E} which is valid for all $\g \le 2$ and all $T>0$.
\end{enumerate}


\section{Appendix: The Fermi symbol for fixed $\rho$ and large $T$}
Our main aim is to derive formula \eqref{rho} for a Hamiltonian $h$ as specified in Subsection \ref{subsect:h}. After a change of variables and replacing $\mu$ with $\mu_\rho(T)$ formula \eqref{eq:density} for the mean particle density takes the form
\begin{align*}
	\varrho\big(T, \mu_\rho(T)\big) = \frac{T^{\frac{d}{2m}}}{(2\pi)^d}\int_{\R^d} \frac{\mathrm{d}\bxi}{1+\exp\big(-\mu_\rho(T)/T\big)\exp \big(h(T^{\frac{1}{2m}}\bxi)/T\big)}\, \,.
\end{align*}
Since $h(T^{\frac{1}{2m}}\bxi)/T\to h_\infty(\bxi)$ as $T\to\infty$, for each $\bxi\not = \bf{0}$, the condition \eqref{rhofixed} requires that $\exp\big(-\mu_\rho(T)/T\big)\to\infty$. Consequently,
\begin{align*}
	\varrho\big(T, \mu_{\rho}(T)\big) = \varkappa\,T^{\frac{d}{2m}}  \exp\big(\mu_\rho(T)/T\big)\big(1+o(1)\big)\,, 
\end{align*}
where $\varkappa$ is defined in \eqref{lambda}. Using \eqref{rhofixed} again leads to \eqref{rho}.

As explained in the Introduction, the high-temperature limit under the condition \eqref{rhofixed} corresponds to the Maxwell--Boltzmann gas limit. This fact can be conveniently restated in terms of the so-called fugacity $z_\rho(T)\ceq\exp\big(\mu_\rho(T)/T\big)$ as follows. By \eqref{eq:asympt} the integrated density of states $\CN(T)$ of \eqref{eq:dos}, satisfies for large $T$ the relation $\CN(T)\asymp T^{\frac{d}{2m}}$. So \eqref{rho} implies $z_\rho(T)\asymp \rho\ / \CN(T)\to 0$. 


\bibliographystyle{beststyle}

\end{document}